\documentclass[12pt]{amsart}
\usepackage{amssymb, latexsym}

\usepackage{amsthm}
\usepackage{amsmath,amscd}
\usepackage{amsfonts}
\usepackage[dvips]{graphicx}

\usepackage[notcite,notref]{showkeys}
\usepackage{amsthm}
\usepackage{amsmath,amscd}
\usepackage{amsfonts}
\usepackage[dvips]{graphicx}

\numberwithin{equation}{section} 
\newcommand{\bean}{\begin{eqnarray*}}
\newcommand{\eean}{\end{eqnarray*}}
\newcommand{\be}{\begin{equation}}

\newcommand{\ee}{\end{equation}}
\newcommand{\bd}{\begin{displaymath}}
\newcommand{\ed}{\end{displaymath}}
\newcommand{\eps}{\varepsilon}

\newcommand{\al}{\alpha}

\newcommand{\bt}{\beta}

\newcommand{\gm}{\gamma}
\newcommand{\Gm}{\Gamma}

\newcommand{\lm}{\lambda}

\newcommand{\om}{\omega}

\newcommand{\cN}{{\mathcal N}}
\newcommand{\cK}{{\mathcal K}}
\newcommand{\cMe}{{\mathcal M}^\eps}

\newcommand{\beq}{\begin{equation}}
\newcommand{\eeq}{\end{equation}}
\newcommand{\bea}{\begin{eqnarray}}
\newcommand{\eea}{\end{eqnarray}}

\newtheorem{lemma}{Lemma}[section]
\newtheorem{theorem}[lemma]{Theorem}
\newcommand{\abs}[1]{\left\vert{#1}\right\vert}

\newcommand{\da}{\dagger}
\newcommand{\reu}{{\rm{Re}}\,\lambda_1}
\newcommand{\imu}{{\rm{Im}}\,\lambda_1}
\newcommand{\R}{\mathbb{R}}
\newtheorem{remark}[lemma]{Remark}
\newcommand{\e}{\varepsilon}
\newcommand{\norm}[1]{\left\Vert#1\right\Vert}
\newcommand{\C}{\mathbb{C}}
\newtheorem{proposition}[lemma]{Proposition}

\begin{document}

\author[J. Rubinstein, P. Sternberg and K. Zumbrun]{Jacob Rubinstein, Peter Sternberg and Kevin Zumbrun}

\address{Department of Mathematics, Indiana University,
Bloomington, IN 47405, USA} \email{jrubinst@indiana.edu,
sternber@indiana.edu, kzumbrun@indiana.edu}

%CHANGED(added date!)-K
%ENDCHANGED
\date{Last updated Dec. 20, 2007}
\thanks{J. Rubinstein was partially supported by the NSF. P. Sternberg was partially supported by
NSF DMS-0401328 and DMS-0654122. Kevin Zumbrun was partially
supported by NSF DMS-0300487}

\title[Resistive state in a superconducting
wire]{ The resistive state in a superconducting
wire:\\
Bifurcation from the normal state}

%\noindent \textbf{Abstract.}
\begin{abstract}
We study formally and rigorously the bifurcation to steady and
time-periodic states in a model for a thin superconducting wire in
the presence of an imposed current. Exploiting the PT-symmetry of
the equations at both the linearized and nonlinear levels, and
taking advantage of the collision of real eigenvalues leading to
complex spectrum, we obtain explicit asymptotic formulas for the
stationary solutions, for the amplitude and period of the
bifurcating periodic solutions and for the location of their zeros
or ``phase slip centers" as they are known in the physics
literature. In so doing, we construct a center manifold for the flow
and give a complete description of the associated finite-dimensional
dynamics.

\end{abstract}
\maketitle

\vskip 0.5cm

\section{Introduction}

One of the natural applications of superconducting is to exploit
their infinite conductivity to transmit electric currents. The goal
of this paper is to analyze a number of asymptotic problems that
arise in the study of such current transmission through a wire. We
consider a simple canonical problem, in which the superconducting
portion of the wire is of a finite extent $-L \leq x \leq L$. It is
assumed that a normal current $I$ is fed into the the wire at its
left end. It is known that under the right circumstances, for
example for a temperature $T$ that is sufficiently small, the
current in the wire itself will be in part normal and in part
superconducting. This coexistence of two types of currents in the
wire is called {\em a resistive state}.

The resistive state in superconducting wires received some attention
by physicists who observed a variety of phenomena that are unique to
this situation. To review these observations, which involve
oscillatory, that is to say, inherently time-varying behavior, it is
standard practice to use the time-dependent Ginzburg-Landau model
(TDGL). For a three-dimensional wire occupying say a thin
cylindrical region $D$ with axis of length $2L$ centered on the
$x$-axis, this system can be written in non-dimensional form as
\begin{eqnarray} & \psi_t+i\phi\psi=
(\nabla-iA)^2\psi+(\Gamma-\abs{\psi}^2)\psi\;\mbox{for}\;(x,y,z)\in
D,\label{3dpsi}\\
& \nabla\times\nabla\times
A=-\sigma(A_t+\nabla\phi)+\frac{i}{2}\left(\psi \psi_x^*-\psi_x
\psi^*\right)-\abs{\psi}^2A \;\mbox{for}\;(x,y,z)\in D,\label{3dA}
\end{eqnarray}
(cf. \cite{tin}) where $\psi:D\to\C$ is the Ginzburg-Landau
parameter whose square modulus measures the density of
superconducting electrons, $A:\R^3\to\R^3$ is the magnetic potential
whose curl measures the effective magnetic field and
$\phi:\R^3\to\R$ is the scalar electric potential whose gradient represents the electric field.
The parameter
$\Gm$ is proportional to $T_c-T$, where $T$ is the actual
temperature, and $T_c$ is the phase transition temperature in the
absence of external currents (i.e. in the case $I=0$). In
\eqref{3dA}, $*$ denotes complex conjugation and the right-hand
side represents the sum of normal current (with associated ohmic
conductivity $\sigma$) and supercurrent. We note that the TDGL is
invariant under the gauge transformation
\[
(\psi,\phi,A)\rightarrow (\psi e^{i g},\phi-g_t,A+\nabla
g)\quad\mbox{for any smooth scalar function}\;g=g(x,y,z,t).\]

 To pursue an appropriate
three-dimensional analysis of the problem of forcing an applied
current into a wire, one would then have to impose inhomogeneous
boundary conditions on the normal component of the normal current
and couple the system above to a Maxwell system on the exterior of
$D$. This is not the direction we will follow; instead we adopt the
model favored for many years in the physics literature on the
subject, e.g. \cite{ivko}, and view the wire as a one-dimensional
object. Before leaving the higher dimensional setting we should
comment, however, that an interesting two-dimensional study of the
stability of the normal state with an applied electric current can
be found in \cite{Almog}.

In such a one-dimensional model, the exterior problem is typically
ignored as a lower-order effect and so all unknowns are taken simply
to be defined along the wire as functions of $x$ and $t$ only.
Furthermore, through the gauge choice $g_x=-A$, one can eliminate
the magnetic potential $A$ completely from the system. Then using
the remaining freedom in the $t$-dependence of $g$, one can insist on the
convenient normalization \beq\phi(0,t)=0\label{phi0}\quad\mbox{for
all}\;t\geq 0.\eeq Under these assumptions and gauge choice,
\eqref{3dpsi} reduces to
\begin{equation}
\psi_t+i\phi \psi=\psi_{xx}+\Gm \psi - |\psi|^2 \psi\quad\mbox{for}\;-L<x<L,\;t>0. \label{gl1}
\end{equation}
Regarding the fate of the second equation \eqref{3dA}, note that
necessarily the divergence of the total current, that is, the
right-hand side, must vanish. In one dimension, this condition
however implies simply that the total current is a constant.
Therefore, since we are specifying that the current at the endpoints
of the wire is purely normal and equal to $I$, we arrive at the
relation
\begin{equation}
\frac{i}{2}\left(\psi \psi_x^*-\psi_x \psi^*\right) - \sigma
\phi_x=I\quad\mbox{for}\;-L<x<L,\;t\geq 0.  \label{gl3}
\end{equation}
Since the natural setup is for the wire to be connected at its
endpoints to a metal exhibiting normal conductivity, we supplement
the system \eqref{phi0}--\eqref{gl3} with homogeneous Dirichlet
boundary conditions on the order parameter
\begin{equation}
\psi(\pm L,t)=0, \label{gl5}
\end{equation}
along with initial conditions on $\psi.$ The general nature of our
results apply also to other homogeneous boundary conditions,
including in particular the case of homogeneous Neumann boundary
conditions.

As long as the temperature is sufficiently high, that is $\Gm$ is
low, the wire is in the normal state. This means that $\psi=0$, and
the current in the wire is purely ohmic, i.e.
\begin{equation}
\phi_x=-I/\sigma. \label{gl7}
\end{equation}
As the temperature $T$ is lowered, one reaches a critical value,
determined by a curve $\Gm=\Gamma_1(I)$ where the normal state loses
its stability and a nontrivial superconducting state might emerge.
Two such states were discovered in the early 1980's. The first one
is a stable stationary solution, reported by Langer and Baratoff
\cite{krba}. The notion of `stationarity' requires some care here.
Let us express the order parameter in polar form $\psi=fe^{i\chi}$;
then the gauge invariant quantities $f(x,t), \; q(x,t):=\chi_x(x,t)$
and $\theta(x,t):=\chi_t(x,t)-\phi(x,t)$ converge to stationary
functions $f_0(x),\; q_0(x),\;\theta_0(x)$. On the other hand,
Langer and Ambegaokar (\cite{laam} (see also Ivlev and Kopnin
\cite{ivko}) found in their numerical simulations that for some
values of the parameters $(I,\Gamma)$, the normal state bifurcates
into a state where the order parameter oscillates in time. This
periodic behavior is particularly interesting in light of the
dissipative nature of the TDGL model (\ref{gl1}), but it is made
possible by the presence of the applied current which effectively
disrupts the gradient-flow structure of the system.

For the duration of this investigation, we fix $L=1$ so that the
wire occupies the interval $[-1,1].$ We also set the conductivity
$\sigma=1$ in order to focus on the different kinds of states that
emerge at different points in the $(I,\Gamma)$ plane.

In a recent work \cite{rsm} the authors used numerical simulations
and some analytical arguments to identify a more elaborate phase
transition picture. In particular the curve $\Gamma_1(I)$ was shown
to be associated with an interesting spectral problem. It was also
shown that there exist two critical currents $I_k$ and $I_c$ that
play an important role in the system behavior. The curve
$\Gamma_1(I)$ is depicted in Figure \ref{1asym}. Specifically, the
normal state (N) is only stable for $\Gamma < \Gamma_1(I)$. When the
temperature is decreased and $\Gamma$ increases past $\Gamma_1(I)$,
the (N) state become unstable. If $I<I_k \approx 10.92$, then the
(N) state bifurcates into a stationary (S) state. On the other hand,
the bifurcation branch to a stationary state is unstable for $I_k <
I < I_c$, while if $I>I_c \approx 12.31$, the (N) state bifurcates
to a stable time-periodic (P) state.

\begin{figure}
\begin{center}
\begin{tabular}{c}
\includegraphics[height=7cm]{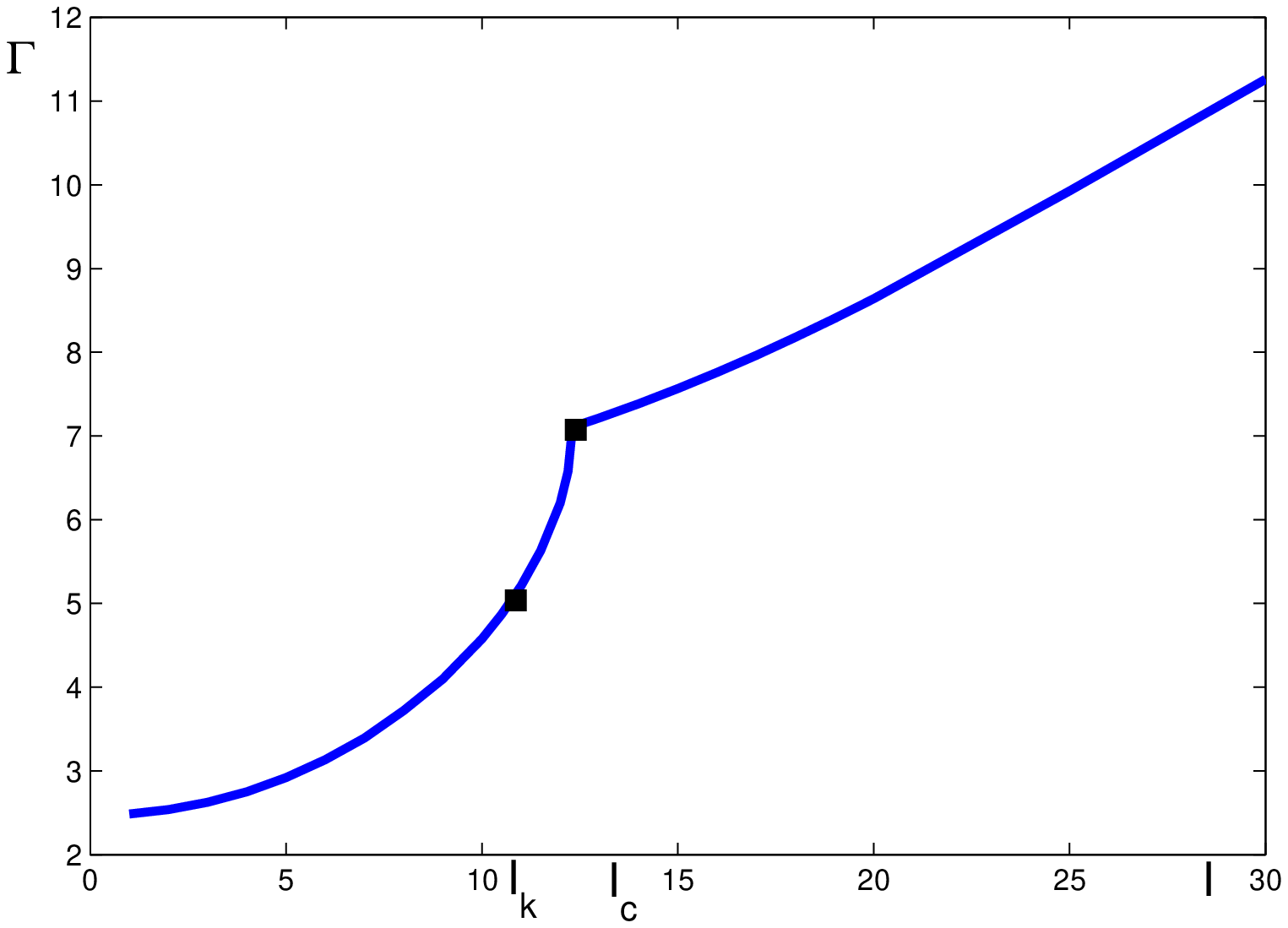}
\end{tabular}
\end{center}
\caption{} \label{1asym}
\end{figure}

To understand the phase transitions described above, and in general
the solution to  equations (\ref{gl1})-(\ref{gl3}), it is useful to
study a number of mathematical problems:
\begin{enumerate}

\item What is the meaning of the special current value $I_{c}$ where
the bifurcating state switches from a stationary one to a periodic
one?

\item From the viewpoint of self-adjoint operators arising in the time-independent
case, the answer to the first question involves an unusual spectral
problem involving the collision of two eigenvalues. Moreover, near
the critical value $I_{c}$, the spectrum of the underlying operator
%CHANGED (OK?)-K:
%changes its nature.
changes its nature, from real to complex.
%ENDCHANGED
Therefore it is desirable to understand this
spectrum near the special value $I_{c}$.

\item What is the nature of the bifurcating branch near the curve
$\Gamma_1(I)$? In the language of dynamical systems, we ask what is
the geometry of the center manifold there? This question, in fact,
involves a number of issues. For instance, is the bifurcation branch
stable and what is its shape? In addition, we point out that in the
periodic case there are points in space-time where the order
parameter $\psi$ vanishes. Such points are called phase slip centers
(PSC's) since the phase `exploits' the vanishing of the amplitude to
have a discontinuity there, thus relaxing large accumulated phase.
Thus one of the relevant questions would be to identify these
points.

\end{enumerate}

The loss of stability of the normal state is studied through the
linearization of equations (\ref{phi0})-(\ref{gl5}) around the
normal state
\begin{equation}
\psi \equiv 0,\;\; \phi=-Ix. \label{gl11}
\end{equation}
It is convenient to express the solution $\psi$ of the linearized
equation in the form $\psi(x,t)=u(x)e^{(\Gm-\lambda) t}$; then we
obtain that $u(x)$ is the solution of the following spectral problem
\begin{equation}
M u:=u_{xx}+ixI u = -\lambda u,\;\;\;u(\pm 1)=0. \label{gl13}
\end{equation}
Clearly the stability of the normal state is determined by whether
$\Gm$ is larger or smaller than ${\rm{Re}}\,\lambda_1$ where
$\lambda_1$ refers to the eigenvalue of the operator $M$ having
smallest real part. In the next section we shall therefore consider
this eigenvalue and examine some of its properties for small $I$ and
for large $I$. The results of this section will help us in
identifying the critical value $I_{c}$. Then, in Section 3, we will
examine in more detail the leading eigenvalue of $M$ for $I$ values
near $I_{c}$. In Section 4 we construct stationary solutions to
\eqref{gl1}-\eqref{gl5} for $I<I_c$ using formal asymptotic
expansions and multiple time-scales. In Section \ref{NP}, we construct
periodic solutions using the same methods in the regime $I>I_c$
where the spectrum of $M$ has become complex. We then make these
calculations rigorous in Section 6 by constructing the center
manifold for the solution immediately after bifurcation and studying
the O.D.E.'s which  govern the flow on the center manifold.

 A novel aspect of the rigorous analysis in Section 6
 is that both the linearized and the
 full nonlinear system admit what is called PT-symmetry; namely an
 invariance under the joint transformation of $x\rightarrow -x$ and
 complex conjugation. This type of symmetry has been the focus of a
 number of recent investigations (see e.g. \cite{bebo,cgs,cdv,latr}).
In particular, in the analysis of periodic bifurcation we make
strong use of this symmetry to reduce the dimension of what turns
out to be a four-dimensional phase space to a planar system
exhibiting a standard Hopf bifurcation. Furthermore, we are able to
then go further and describe bifurcation for the full
four-dimensional system involving possibly non-PT-symmetric
solutions, proving that, at least in the vicinity of the normal
solution $\psi\equiv 0$ and for values of the bifurcation parameter
that we study,
%CHANGED (wasn't quite correct!)-K:
%all solutions eventually converge, up to fixed
%complex rotation, to the manifold of solutions exhibiting PT-symmetry.
solutions generically converge, up to fixed
complex rotation, to the manifold of solutions exhibiting PT-symmetry,
and thereafter to the stable, PT-symmetric periodic solutions arising through
Hopf bifurcation within that manifold.
%par. break here?
However (Theorem \ref{Hopf2}), we also exhibit through direct calculation
the existence of unstable, non-PT-symmetric
periodic solutions in the same vicinity, i.e., persistent solutions
that do not converge to the manifold of PT-symmetry.
This shows that the observed generic convergence to PT-symmetry
is the result of detailed local dynamics on the center manifold about
the normal state and not, say, a global principle associated with
a decreasing Lyapunov functional.
In particular, there might exist attracting steady or periodic states
far from the normal state that are not PT-symmetric, an intriguing
possibility to keep in mind in further investigations.
%That is, convergence to PT-symmetry is the rule but not the law.
%ENDCHANGED

Finally, in Section 7, we show that our rigorous construction of
periodic solutions leads to a proof of the appearance of the phase
slip centers, that is, periodically appearing zeros of the order
parameter.

\section{The spectrum of the canonical PT-symmetric problem for
small $I$ and for  large $I$}

In this section we examine the spectrum of $M$ defined in equation
\eqref{gl13}. The operator $M$ is not self-adjoint, of course. On
the other hand it enjoys a symmetry that is called PT. The letter P
stand for parity, i.e. transforming $x \rightarrow -x$, while the
letter $T$ stands for time reversal, i.e. complex conjugacy. One
readily checks that under this pair of operations, $M$ is unchanged.

PT-symmetric spectral problems seem to have been little studied
until quite recently. In one of the earliest works on this subject
(in a physics context), Bender and Boettcher \cite{bebo} considered
a canonical PT-symmetric operator on the entire real line and
observed through numerical simulations that the spectrum is real. In
the case of a finite interval as in equation \eqref{gl13} above, the
situation is more involved. In particular, we will formally
demonstrate the appearance of complex eigenvalues for $I$ large with
corresponding eigenfunctions possessing an internal layer.

We look first in the case of small $I$. When $I=0$, the spectrum is
of course real, and can, in fact, be written down explicitly
\begin{equation}
\lambda_k(I=0)=\left(\pi k/2\right)^2,\;\;\; k=1,2,... \label{m2b}
\end{equation}

The PT-symmetry of $M$ implies that if $(\lambda,u(x))$ is a
spectral pair of eigenvalue and eigenfunction, then
$(\lambda^*,u^*(-x))$ is also a spectral pair. Langer and Tretter
\cite{latr} have shown that, as long as the eigenvalues of a
PT-symmetric problem are simple, the spectrum is a smooth function
of the problem's parameters. In our case it implies that since the
eigenvalues are well-separated at $I=0$, they are smooth functions
of $I$ at least for $I$ small. However, this implies that the
eigenvalues must remain real also for $I$ positive (but small),
since a real eigenvalue can become complex only by splitting into a
pair of eigenvalues (by the PT-symmetry).

What happens when $I$ is increased? As long as the eigenvalues do
not collide, they remain real. We now show formally that in fact the
eigenvalues must collide for some $I$ and establish an asymptotic
formula for the leading (complex) eigenvalue when $I$ is large.
Precise and rigorous asymptotics for this spectral problem were
carried out by Shkalikov \cite{shk1}. His work (see also
\cite{shk2}) was performed in the context of the Orr-Somerfeld
equation in hydrodynamic stability theory and makes use of a change
of variables leading to an Airy-type equation. We present the
following formal calculation with the hope that it makes the
emergence of the complex spectrum more understandable. Note also
that the formal method given here applies in more general
circumstances, whereas the exact solution of \cite{shk1} is specific
to the exact form of the equations under consideration.

To study the spectrum as $I$ becomes large, it is convenient to
introduce a small positive parameter $\eps$, and then write
$I=\eps^{-2}$. It is clear that the eigenvalues must also be large
to balance the large `potential' $i\eps^{-2}x$. We therefore write
to leading order
\begin{equation}
\lambda=\eps^{-2}(\alpha + i \beta)+o(\eps^{-2}). \label{m3}
\end{equation}
The eigenvalue problem \eqref{gl13} can be written to leading order
as $u_{xx}+\eps^{-2}{\mathcal Q}u=0$, where ${\mathcal
Q}:=x+\beta-i\al$.
Crudely setting $\eps=0$, we obtain a formal eigenvalue problem
$(ix-(\al +i\beta))\psi=0$ for the multiplication operator
$ix$, which evidently has only pure imaginary, essential spectrum.
From this we may conjecture that the spectrum of $L$ becomes complex
as $I\to \infty$ ($\eps\to 0$); however, this limit is very singular
and must be examined in more detail (indeed, on the whole line,
the spectrum is real for $I$ large \cite{ivko}.)

Consider first the case where the potential
${\mathcal Q}$  does not vanish for $x$ in the interval $[-1,1]$.
For example, this would occur if $\alpha \neq 0$. If ${\mathcal Q}
\neq 0$ then there is no turning point in a standard JBKW
\footnote{We use the term JBKW instead of WKB since this expansion
method was introduced by Jeffries in 1923, three years before it was
rediscovered by Wentzel, Kramers and Brillouin, who are also now
ordered alphabetically.} expansion of equation \eqref{gl13}.
Therefore, we seek in this case an asymptotic expansion of the form
\begin{equation}
u(x)=e^{iS(x)/\eps}. \label{m7}
\end{equation}
Substituting (\ref{m7}) and (\ref{m3}) into \eqref{gl13} we get to
leading order
\begin{equation}
S_x=\pm i^{1/2}\left(x+\beta-i\al \right)^{1/2}. \label{m9}
\end{equation}

Integrating the last equation from $-1$ to $x$ gives
\begin{equation}
S^{\pm} =\pm \frac{2}{3}i^{1/2}\left((x+\beta-i\al)^{3/2}-
(-1+\beta-i\al)^{3/2} \right). \label{m11}
\end{equation}
For future reference we introduce the notation
\begin{equation}
S^+(1)={\rm{Re}}\,S+i{\rm{Im}}\,S.  \label{m13}
\end{equation}

The general solution to the equation \eqref{gl13} is (to leading
order)
\begin{equation}
u(x) \sim A e^{iS^+(x)/\eps} +B e^{iS^-(x)/\eps}. \label{m15}
\end{equation}
Substituting this solution into the boundary conditions at $\pm 1$,
and seeking a pair of nontrivial coefficients $A,B$ leads to the
following complex-valued equation
\begin{equation}
e^{iS^+(1)/\eps}=e^{iS^-(1)/\eps}. \label{m17}
\end{equation}
In particular the magnitudes of the two sides of equation
(\ref{m17}) must be the same. Equating the absolute values, and
using the notation (\ref{m13}) gives
\begin{equation}
e^{{\rm{Im}}\,S/\eps}=e^{-{\rm{Im}}\,S/\eps}. \label{m19}
\end{equation}
Therefore, a regular JBKW expansion without a turning point is
feasible only if ${\rm{Im}}\,S=0$.

Recalling the definition of ${\rm{Im}}\,S$, the last condition on it
implies (for some real number $\chi$)
\begin{equation}
(1+\bt-i\al)^{3/2}- (-1+\bt-i\al)^{3/2}  = \chi i^{-1/2}.
\label{m21}
\end{equation}
We now show that equation (\ref{m21}) holds for any $\alpha$ if
$\bt=0$. Setting $\bt=0$, we write
\begin{equation}
-1-i \alpha = \rho e^{i(-\pi+\mu)}, \;\; 1-i \alpha = \rho
e^{i(-\mu)}. \label{s41}
\end{equation}
Here $\rho=\sqrt{1+\al^2}$. Substituting this into (\ref{m21}) with
$\bt=0$, and defining $\gm=3\mu/2$ gives
\begin{equation}
\rho^{3/2}i^{-1/2}G=\chi i^{1/2},\;\;\;G=\sqrt{2}\left(\cos \gm
+\sin \gm\right). \label{s43}
\end{equation}
This proves the assertion above with
\begin{equation}
\chi= \sqrt{2}(1+\al^2)^{3/2}\left(\cos \gm +\sin \gm\right).
\label{s45}
\end{equation}

It remains to show that there exist values of $\alpha$ for which
equation (\ref{m17}) is solvable. Since equation (\ref{m17}) holds
if and only if ${\rm{Im}}\,S=0$ and $\sin\left({\rm{Re}}\,S
\right)=0$, we obtain the condition
\begin{equation}
{\rm{Re}}\,S=\eps n \pi, \label{s47}
\end{equation}
where $n$ is an integer. The calculation above for $\chi$ together
with equation (\ref{s47}) imply
\begin{equation}
\frac{2\sqrt{2}}{3}\left(1+\al^2\right)^{3/2}\left(\cos \gm + \sin
\gm\right)=\eps n \pi. \label{s49}
\end{equation}
Consider the left-hand side as a function of $\al$ (recall that
$\gm$ also depends on $\alpha$ through the relation (\ref{s41})).
When $\al$ tends to $\infty$ then $\gm$ tends to $3\pi/4$, and the
term $\cos \gm + \sin \gm$ on left-hand side approaches zero.
However, the $\al^3$ term on the left hand side grows faster, and
therefore the entire left-hand side goes to $\infty$. On the other
hand, when $\alpha=0$, the left-hand side becomes $2\sqrt{2}/3$.
Therefore, varying $\alpha$ the left-hand side obtains all the
values in the interval $(2\sqrt{2}/3,\infty)$. This means that for
any fixed $\eps$ there exist infinitely many $n$ values for which
equation (\ref{s49}) has a solution $\al_n$. Consequently there are
infinitely many real eigenvalues of order $O(\eps^{-2})$.

The JBKW expansion above captured the real eigenvalues. They are all
of $O(\eps^{-2})$. However, this expansion fails when $\alpha=0$ and
$-1 \leq \beta \leq 1$ since in this case there is a turning point.
We shall now construct a solution for such a case. If we order the
eigenvalues by their real part, the eigenvalues we shall now
construct come before those derived above. Typically a JBKW
expansion captures the large eigenvalues and the associated
oscillatory behavior of the eigenfunction. The lower eigenfunctions
tend to oscillate less. Since we are now looking into the possible
construction of a complex eigenvalue, we need to recall that they
come in pairs of conjugate numbers. Geometrically it means that we
anticipate two eigenfunctions, related by PT-symmetry, and therefore
we anticipate in one case a turning point in $[-1,0)$ and in another
case a turning point in $(0,1]$, i.e. the turning point is at $x=\pm
\beta$. While we seek now an eigenvalue that to leading order is
purely imaginary, we shall be able to obtain also a lower order
correction for it that will have real and imaginary parts. Therefore
we assume that the eigenvalue is of the form
\begin{equation}
\lambda \sim i\eps^{-2} \beta_0 + \eps^{-\nu} \left(\alpha_1
+i\beta_1\right), \label{s71}
\end{equation}
where the exponent $\nu$ is still to be determined.

\begin{figure}
\begin{center}
\begin{tabular}{c}
\includegraphics[height=7cm]{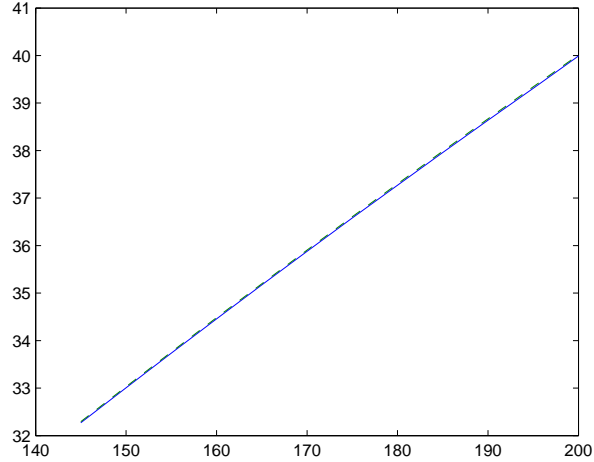}
\end{tabular}
\end{center}
\caption{The (numerically computed) real part of the leading
eigenvalue (solid line) compared to the the asymptotic expansion
(\ref{s83} (dashed line). The horizontal axis in this figure is the
current $I$, and the vertical axis is the real component of the
eigenvalue.} \label{largereallam}
\end{figure}

Consider, then, without loss of generality, the case where the
turning point is in the interior of the interval $(-1,0)$, i.e. $-1
< \beta_0 < 0$. To leading order the potential is balanced exactly
at the point $x=-\beta_0$. We therefore expect the eigenfunction to
be supported in a small neighborhood of this point, and to decay
away from it. Hence, we construct an internal layer around this
point by defining an inner variable $y$ through
\begin{equation}
x=-\beta_0 + \eps^{\gamma} y \label{s73}
\end{equation}
for some $\gm >0$. Substituting the transformation (\ref{s73}) into
equation \eqref{gl13} we see that in order to balance the different
terms in the equation we must set $\gamma=2/3$ and $\nu=4/3$. We
thus obtain on the $y$ scale the balanced equation
\begin{equation}
u_{yy}+ iyu+\left(\alpha_1 + i\beta_1\right)u=0. \label{s75}
\end{equation}
While equation (\ref{s75}) describes the internal layer form of the
eigenfunction $u$, the outer solution is of course $u \equiv 0$.
Therefore, matching the inner and outer solutions implies that we
should consider equation (\ref{s75}) over the entire real line with
the conditions
\begin{equation}
u(y=\pm \infty)=0. \label{s77}
\end{equation}

The eigenvalue problem (\ref{s75})-(\ref{s77}) can be solved
explicitly in terms of Bessel (or Hankel) functions. This was done
by Ivlev and Kopnin \cite{ivko} who concluded that this problem has
no solution with finite $L_2$ norm. This means that our assumption
that $\beta_0 < 1$ is not consistent. In other words, the
concentration cannot occur in the interior of $[-1,1]$. We therefore
consider now the last remaining case in which $\beta_0=1$. The
scaling (\ref{s73}) of the internal variable $y$ is the same, except
that now it is more appropriate to call it the `boundary layer'
variable. Thus, the boundary layer equation (\ref{s75}) is
considered over the half line $y\in [0,\infty)$. The Dirichlet
condition and the matching to the outer solution together imply the
condition
\begin{equation}
u(0)=u(\infty)=0. \label{s79}
\end{equation}
The half-line eigenvalue problem (\ref{s75}), (\ref{s79}) was also
studied in \cite{ivko}. In this case the eigenfunction has finite
$L_2$ norm. The authors computed the leading eigenvalue to be
approximately
\begin{equation}
\alpha_1 + i\beta_1 = 1.17 - 2.02 i. \label{s81}
\end{equation}
Returning to the original notation for the current, we obtain the
eigenvalue
\begin{equation}
\lambda  \sim  1.17 I^{2/3} + i\left(I-2.02 I^{2/3} \right) \;\;\;
{\rm for} \;\; I \gg 1. \label{s83}
\end{equation}

In Figure \ref{largereallam} we depict the asymptotic expansions for
the real part of the leading eigenvalue (dashed line) and the
numerically computed real part (solid lines). Similarly, we depict
in Figure \ref{largeimaglam} the asymptotic imaginary part (dashed
line) and the actual imaginary part (solid line). In both cases the
curves are very close to each other (the error is roughly 0.02).

\begin{figure}
\begin{center}
\begin{tabular}{c}
\includegraphics[height=7cm]{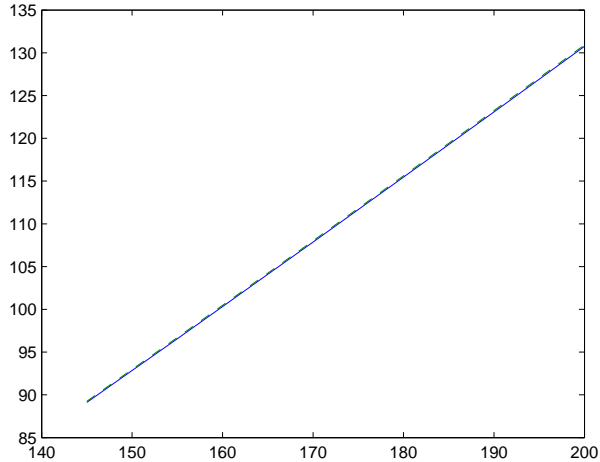}
\end{tabular}
\end{center}
\caption{The (numerically computed) imaginary part of the leading
eigenvalue (solid line) compared to the asymptotic expansion
(\ref{s83} (dashed line). The horizontal axis in this figure is the
current $I$, and the vertical axis is the imaginary component of the
eigenvalue.} \label{largeimaglam}
\end{figure}

The eigenvalues are arranged by their real part in an ascending
order. It is found numerically that as $I$ increases towards the
critical value of $I_c \approx 2.27$ (for $L=1$ and Neumann boundary
conditions) and at the critical value $I_c \approx 12.31$ (for $L=1$
and Dirichlet boundary conditions), the first and second eigenvalue
collide. These results are consistent with the bounds of ref.
\cite{latr}. In the Neumann case their estimate for $I$ below which
the entire spectrum is real is $I<\pi^2/8$, while the corresponding
Dirichlet estimate is $I<3\pi^2/8$.

\section{Eigenvalue collision in PT-symmetric problems}

In this section we look in some detail at the collision process of
two eigenvalues. Our goal is to derive an asymptotic expansion for
the eigenvalues and eigenfunctions near collision, using $I-I_c$ as
our small parameter. The analysis in this section will again be
formal.

For this purpose let $I$ be near a critical value $I_{c}$ the first
two eigenvalues (ordered by their real parts) $\lambda_1$ and
$\lambda_{2}$ collide; that is, the eigenvalues are real for $I$
just below $I_{c}$ but coincide at $I=I_c$.  Let the associated
eigenfunctions be $u_{1}$ and $u_{2}$, respectively. We should note
that the analysis we will present below is valid near the collision
of {\it any} two eigenvalues but we focus on the first collision
since this is most relevant to the stable bifurcation picture to be
presented subsequently.

It follows from equation \eqref{gl13} that they satisfy
\begin{equation}
\int_{-1}^1 u_{1}(x) u_{2}(x) \; dx=0. \label{s9}
\end{equation}
We term this property as $PT-orthogonality$. It is different, of
course, from usual orthogonality. As long as the eigenvalues are
real, their associated eigenfunctions are (up to a normalization)
PT-symmetric. Therefore PT-orthogonality is the same as
orthogonality in the Krein inner product: $[f,g]=\int
f(x)g^*(-x)\;dx=0$ \cite{latr}.

Now, as $I$ approaches $I_{c}$, we assume that $\lambda_i$ and
$\lambda_{ii}$ approach a common value $\lambda^{(0)}$. At the same
time, $u_{1}$ and $u_{2}$ also approach a common function that we
call $u^{(0)}$. Notice, that this statement is made after some
proper normalization, since the problem is linear. We really mean
that the ratio $u_{1}/u_{2}$ approaches a complex constant. To see
why, assume to the contrary that there are two independent
eigenfunctions $u_{1}$ and $u_{2}$ associated with the real
eigenvalue $\lambda^{(0)}$. Equation \eqref{gl13} is of second
order, and therefore its solution space is spanned by two
independent functions. However, $u_{1}$ and $u_{2}$ cannot form such
a basis, since they both satisfy homogeneous Dirichlet boundary
condition, while clearly there are solutions of \eqref{gl13} that do
not satisfy such conditions.

Having established that the collision eigenvalue $\lambda^{(0)}$ has
an algebraic multiplicity $2$, but a geometric multiplicity $1$ (and
thus we can say that the operator $M$ is Jordan at $I_{c}$), we
proceed to study what happens when $I$ increases past $I_{c}$.

We first write the eigenvalue problem for the operator $M$ near
$I_{c}$ in the form
\begin{equation}
u_{xx} + i(I_{c}+\eps a)xu +\lambda u=0, \label{s11}
\end{equation}
where the sign of the parameter $a$ determines if we move up or down
from $I_{c}$. Because of the singular nature of $M$ at $I_{c}$, it
turns out that the perturbation scheme is not analytic. Rather, we
need to expand the eigenfunction $u$ and the eigenvalue $\lambda$ in
powers of $\eps^{1/2}$:
\begin{equation}
u=u^{(0)}+\eps^{1/2} u^{(1)}+ \eps u^{(2)} + ....,\;\;\;\;
\lambda=\lambda^{(0)}+ \eps^{1/2} \lambda^{(1)} + \eps \lambda^{(2)}
+ ....\label{s13}
\end{equation}
At the first order we find of course
\begin{equation}
{\mathcal L}u^{(0)}:=u^{(0)}_{xx} + ixI_{c} u^{(0)}
+\lambda^{(0)}u^{(0)}=0, \label{s15}
\end{equation}
where we used this opportunity to introduce the operator notation
${\mathcal L}$. It is important to note that the PT-orthogonality
noted above implies
\begin{equation}
\int_{-1}^1 u^{(0)}(x)^2 \; dx=0. \label{s17}
\end{equation}
At the $O(\eps^{1/2})$ level we get
\begin{equation}
{\mathcal L}u^{(1)}=u^{(1)}_{xx}  + iI_{c} xu^{(1)} +\lambda^{(0)}
u^{(1)}=-\lambda^{(1)} u^{(0)}. \label{s19}
\end{equation}
Multiplying the last equation by $u^{(0)}$, integrating over the
interval $[-1,1]$, and using the PT-orthogonality (\ref{s17}) we see
that equation (\ref{s19}) is solvable. However, unlike the case of
regular eigenvalue perturbation schemes, we gain no information on
$\lambda^{(1)}$ at this level. We therefore need to proceed to the
$O(\eps)$ level:
\begin{equation}
{\mathcal L}u^{(2)}=u^{(2)}_{xx} + iI_{c} xu^{(2)} +\lambda^{(0)}
u^{(2)}=-\lambda^{(1)} u^{(1)} - \lambda^{(2)} u^{(0)} -iaxu^{(0)}.
\label{s21}
\end{equation}
To get a solvability condition we multiply both sides by $u^{(0)}$,
integrate over the interval and use (\ref{s17}) to find
\begin{equation}
-\lambda^{(1)} \int_{-1}^1 u^{(1)}(x) u^{(0)}(x) \; dx = ia
\int_{-1}^1 x u^{(0)}(x)^2\;dx \label{19}
\end{equation}
It is convenient at this point to introduce some notation. First, we
set
\begin{equation}
u^{(0)}={\rm{Re}}\,u^{(0)}(x)+i{\rm{Im}}\,u^{(0)}(x).  \label{21}
\end{equation}
Using the PT symmetry of $u^{(0)}$, we choose a normalization in
which ${\rm{Re}}\,u^{(0)}$ is even, while ${\rm{Im}}\,u^{(0)}$ is
odd. Then we define the real parameter $a_1$ through
\begin{equation}
a_1=-i \int_{-1}^1 xu^{(0)}(x)^2\;dx = 2\int_{-1}^1 x
{\rm{Re}}\,u^{(0)}(x) {\rm{Im}}\,u^{(0)}(x) \;dx \label{s23}
\end{equation}
Next, let $K(x)$ be the solution of the nonhomogeneous ODE
\begin{equation}
{\mathcal L}K= K_{xx} + iI_c xK +\lambda^{(0)}K = u^{(0)},\;\;\;\;
K(\pm 1)=0. \label{s25}
\end{equation}
The identity (\ref{s17}) ensures that equation (\ref{s25}) is
solvable. Using this canonical function $K$, we express $u^{(1)}$ as
\begin{equation}
u^{(1)}(x)=-\lambda^{(1)} K(x). \label{s27}
\end{equation}
Finally, we define
\begin{equation}
b=\int_{-1}^1 K(x) u^{(0)}(x) \; dx. \label{s29}
\end{equation}
Notice that $K$ is also PT-symmetric, i.e. $K(x)=\bar{K}(-x)$. In
particular ${\rm{Re}}\,K$ is even while ${\rm{Im}}\,K$ is odd. A
numerical integration of $K$, and a numerical evaluation of the
functionals in (\ref{s23}) and (\ref{s29}) gives
\begin{equation}
a_1 \approx 0.29,\;\;\; b \approx 0.12. \label{s30}
\end{equation}

We use the notation above to derive from equation (\ref{s19}) the
relation
\begin{equation}
\lambda_1^2 = -a a_1 /b  \approx -2.42 a. \label{s31}
\end{equation}
When $a<0$, $I$ is just below $I_{c}$ and there are two real
solutions. The negative one corresponds to the first eigenvalue, and
the positive one corresponds to the second eigenvalue. On the other
hand, when $I$ increases past $I_{c}$, i.e. when $a>0$, there is a
complex pair of conjugate solutions. This implies that the critical
eigenvalue $\lambda^{(0)}$ splits into a complex conjugate pair with
the $\mathcal{O}(\e^{1/2})$ correction $\lambda^{(1)}$ being purely
imaginary. We remark that the same analysis applies to any collision
of real eigenvalues. In Figure \ref{imagelam} we compare the
asymptotic expansion (\ref{s13}) of the imaginary part of the first
eigenvalue (dashed line) with the numerically computed value (solid
line).

\begin{figure}
\begin{center}
\begin{tabular}{c}
\includegraphics[height=7cm]{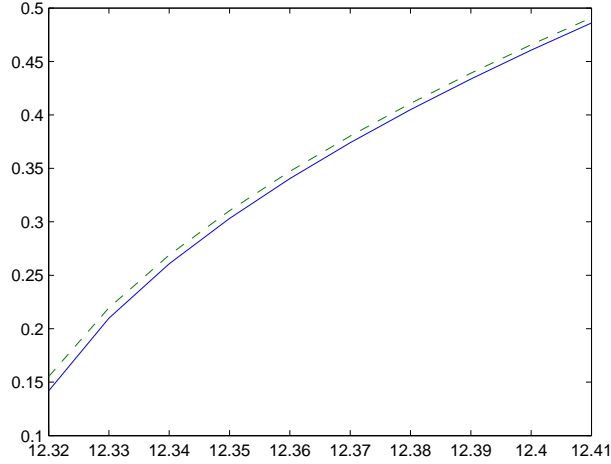}
\end{tabular}
\end{center}
\caption{The (numerically computed) imaginary part of the leading
eigenvalue (solid line) compared to the leading order term in the
asymptotic expansion (\ref{s13} (dashed line). The horizontal axis
in this figure is the current $I$, and the vertical axis is the
imaginary component of the eigenvalue.} \label{imagelam}
\end{figure}

To find the next term $\lambda^{(2)}$ in the eigenvalue expansion,
we proceed further to the $O(\eps^{3/2})$ level:
\begin{equation}
{\mathcal L} u^{(3)} =-\lambda^{(2)} u^{(1)} -\lambda^{(1)} u^{(2)}-
\lambda^{(3)} u^{(0)} -iaxu^{(1)}. \label{s61}
\end{equation}
Multiply equation (\ref{s61}) by $u^{(0)}$ and integrate by parts
over $[-1,1]$ to get
\begin{equation}
\lambda_2\int_{-1}^1 u^{(1)} u^{(0)}\,dx + \lambda^{(1)}\int_{-1}^1
u^{(2)} u^{(0)}\,dx + ia \int_{-1}^1 x u^{(1)} u^{(0)}\,dx=0.
\label{s63}
\end{equation}
Substituting the relation (\ref{s27}) into equation (\ref{s63}) and
dividing by $\lambda^{(1)}$ gives
\begin{equation}
-b \lambda^{(2)}  + \int_{-1}^1 u^{(2)} u^{(0)}\,dx -i a \int_{-1}^1
xK(x)u^{(0)}\,dx=0. \label{s65}
\end{equation}
It is useful at this point to introduce additional canonical
functions and functionals in the spirit of $K$, $a_1$ and $b$
defined above. Thus we define two canonical functions $\zeta$ and
$w$ through:
\begin{equation}
{\mathcal L}\zeta = K(x)-  \theta_1.\;\;\;\zeta(\pm 1)=0.
\label{h43}
\end{equation}
Here $\theta_1$ is a constant chosen such that equation (\ref{h43})
is solvable. Namely
\begin{equation}
\theta_1=\int_{-1}^1 K(x) u^{(0)}(x)dx /\int_{-1}^1 u^{(0)}(x)dx.
\label{h43b}
\end{equation}
The next canonical function $w(x)$ is defined by
\begin{equation}
{\mathcal L}w=-ixu^{(0)} - \theta_2,\;\;\;w(\pm 1)=0, \label{h45}
\end{equation}
with
\begin{equation}
\theta_2=-i\int_{-1}^1 x u^{(0)}(x)^2\,dx /\int_{-1}^1 u^{(0)}(x)dx.
\label{h45b}
\end{equation}

Using $\zeta$, $w$, and $K$ we can write
\begin{equation}
u^{(2)}(x)=(\lambda^{(1)})^2 \zeta(x) + a w(x) +\lambda^{(2)} K(x).
\label{s64}
\end{equation}
We further define the functionals
\begin{equation}
d_1=\int_{-1}^1 \zeta u^{(0)}\;dx, \;\; d_2=\int_{-1}^1
wu^{(0)}\;dx, \;\;d_3=i\int_{-1}^1 xK(x)u^{(0)}(x)\,dx, \label{h45c}
\end{equation}
A numerical computation gives
\begin{equation}
d_1 \approx -0.014,\;\; d_2 \approx -0.02,\;\;d_3 \approx -0.02.
\label{h45d}
\end{equation}
Using these functionals and formula (\ref{s64}) for $u^{(2)}$ in
equation (\ref{s65}), and using (\ref{s31}) to eliminate the
contribution of the coefficients $\theta_1,\theta_2$, gives the
following expression for $\lambda_2$:
\begin{equation}
\lambda^{(2)} = \left(ad_3-ad_2+(\lambda^{(1)})^2 d_1 \right)/2b.
\label{s67}
\end{equation}

We can conclude now an interesting fact. The eigenvalue
$\lambda^{(0)}$ splits into two eigenvalues as the current is varied
away from $I_{c}$. For $I<I_{c}$ we obtain a real pair, while for
$I>I_{c}$ we obtain a complex pair. This splitting manifests itself
in the two values for $\lambda^{(1)}$ obtained from equation
(\ref{s31}). We can also see how the single eigenfunction $u^{(0)}$
splits into two eigenfunctions through equation (\ref{s19}). On the
other hand, $\lambda^{(1)}$ appears in equation (\ref{s67}) only
through its square. Therefore equation (\ref{s67}) implies that
$\lambda^{(2)}$ is unique and real. In particular, if we draw the
real part of the colliding eigenvalues as a function of $I$ near the
collision, we obtain that the function is not analytic at $I_{c}^-$.
In fact $d\lambda_0/dI$ blows up as we approach $I_{c}$ from below,
due to the $\mathcal{O}(\e^{1/2})$ contribution to the expansion
\eqref{s13} coming from \eqref{s31}. Yet due to the fact that
$\lambda^{(1)}$ is purely imaginary for $I$ just above $I_c$, we see
that the graph of
${\rm{Re}}\,\lambda_1(I)\;(={\rm{Re}}\,\lambda_2(I))$ is
differentiable from the right at $I=I_c.$ This analytical conclusion
is verified in the numerical solution.

In Figure \ref{reallam} we compare the asymptotic expansion
(\ref{s13}) of the real component of the first eigenvalue (dashed
line) with the numerically computed value (solid line). The two
lines are almost indistinguishable, and the error is $O(0.001)$.

\begin{figure}
\begin{center}
\begin{tabular}{c}
\includegraphics[height=7cm]{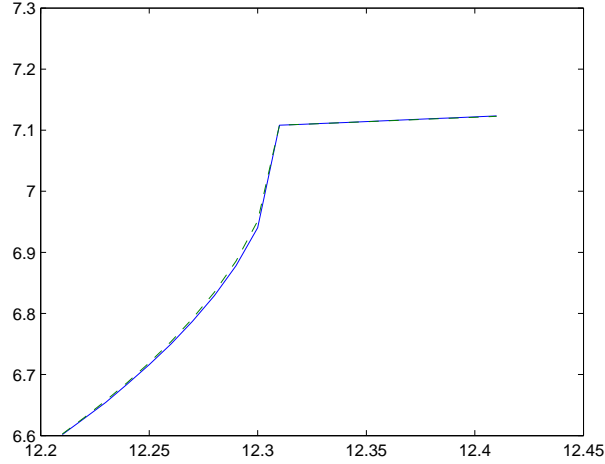}
\end{tabular}
\end{center}
\caption{The (numerically computed) real part of the leading
eigenvalue (solid line) compared to the leading order term in the
asymptotic expansion (\ref{s13} (dashed line). The horizontal axis
in this figure is the current $I$, and the vertical axis is the real
component of the eigenvalue.} \label{reallam}
\end{figure}

\section{The bifurcation from (N) to (S)}

In the next two sections we study the shape and the stability of the
bifurcation branch from the normal (N) state to either the
stationary (S) or periodic (P) state using formal asymptotic
expansions and multiple time-scales. Later, in Section 6, we will
present a rigorous justification of these calculations by appealing
to center manifold theory.

 For this purpose, it is convenient to rewrite
the system \eqref{gl1}-\eqref{gl3} as a single nonlocal complex
equation by first solving for the electric potential $\phi$ in
\eqref{gl3}, thereby obtaining \beq \phi=-Ix+\frac
i2\int_0^x\left(\psi\psi^*_x-\psi^*\psi_x\right)\,dx'.\label{solvephi}\eeq
Then we substitute this into \eqref{gl1} to obtain \beq \psi_t=
\psi_{xx}+ixI\psi+\Gamma\psi+\mathcal{N}[\psi],\label{reforma}\eeq
where we have introduced notation for the cubic nonlinearity \beq
\mathcal{N}[\psi]:=-\abs{\psi}^2 \psi+ \frac 12\;
\psi\int_0^x\left(\psi \psi^*_x-\psi^*
\psi_x\right)\,dx'.\label{nonnot}\eeq As always, this equation is
augmented with Dirichlet boundary conditions at $x=\pm 1$ and
initial conditions.

 We work in this section with a fixed current $I$ in the
regime $I<I_c$ where the leading eigenvalue $\lambda_1$ of the
operator $M$ is real (as is the entire spectrum) and we expect a
bifurcation to a stationary state. The transition takes place
exactly when $\Gamma$ crosses the value $\lambda_1$. This defines
the curve $\Gamma = \Gamma_1(I)$.

We denote the leading eigenfunction by $u_1$, and normalize it as
usual by $u_1(0)=1$. Recall that $u_1$ satisfies the equation \beq
L_1 u_1:=(u_1)_{xx}+ixIu_1+\lm_1u_1=0.\label{Lone}\eeq

To find the solution to \eqref{reforma} just above the curve
$\Gamma_1(I)$, we set $\Gm=\lm_1 +\eps$, so that \eqref{reforma}
takes the form \beq
\psi_t=L_1\psi+\e\psi+\mathcal{N}[\psi].\label{reformb}\eeq

Anticipating the contribution of the nonlinear terms in the
forthcoming expansion, we seek a solution that is proportional to
leading order to $u_1$ plus a small perturbation:
\begin{equation}
\psi(x,t) \sim \eps^{1/2} \al(\tau) u_{1}(x) +\eps^{3/2}
\psi_1(x,\tau) +... \label{h279}
\end{equation}
Since $\Gamma$ is a small perturbation of $\lambda_{1}$, we expect
the time evolution to be slow, hence we have introduced the
time-scale $\tau=\eps t$. Our goal here is to compute the function
$\al(\tau)$ and thus to obtain completely the leading order term in
the expansion.

Substituting the ansatz (\ref{h279}) into \eqref{reformb}, we see
through \eqref{Lone} that the $O(\eps^{1/2})$ terms are balanced by
the choice above for $\psi$, with the function $\al(\tau)$ not yet
determined. Proceeding the $O(\eps^{3/2})$ level, we obtain
\begin{equation}
-L_{1}\psi_{1} = \left(-\al_{\tau}+\al\right) u_{1} +
\mathcal{N}[\al u_{1}]. \label{h281}
\end{equation}

To obtain a solvability condition for $\psi_1$, we multiply equation
(\ref{h281}) by $u_{1}(x)$ and integrate over $[-1,1]$. We obtain
the following equation for $\al(\tau)$:
\begin{equation}
\al_{\tau}=\al+\frac{\int_{-1}^{1}\mathcal{N}[\al
u_{1}]u_{1}\,dx}{\int_{-1}^{1}u_{1}^{2}\,dx} =\al+\chi_{11} |\al|^2
\al. \label{h285}
\end{equation}
After a lengthy calculation, the coefficient $\chi_{11}$ is found to
be given by
\begin{equation}
\chi_{11} = \left(\frac 12 c_{1111}-\gm_{11}\right)/\bt,
\label{h287}
\end{equation}
where
\begin{equation}
\bt=\int_{-1}^1 u_{1}^2\;dx,\;\;\gm_{11}= \int_{-1}^1 |u_{1}|^2
u_{1}^2,\;\; c_{1111}=\int_{-1}^1 u_{1}^2 \theta_{11} \;dx,
\;\;\theta_{11}=\int_{0}^x u_{1} (u_{1})_x^* - u_{1}^*
(u_{1})_x\,dx'. \label{h289}
\end{equation}

The stability of the bifurcation branch depends on the sign of
$\chi_{11}$. Examination of \eqref{h285} reveals that there is
stable branch of equilibria if $\chi_{11} <0$, with
$\abs{\al}=\frac{1}{\sqrt{-\chi_{11}}}$. We note that in light of
the rotational invariance of the whole problem \eqref{reforma},
there is in fact an entire circle of equilibria with $\al$ given by
$\frac{1}{\sqrt{-\chi_{11}}}e^{i\theta_{0}},\;\theta_{0}\in
[0,2\pi)$.

On the other hand, if $\chi_{11}>0$ then an unstable branch of
equilibria exists for $\e$ small and negative wherein \eqref{h285}
is replaced by
\[\al_{\tau}= -\al+\chi_{11} |\al|^2 \al\] and the equilibrium value
of $\al$ is given by $1/\sqrt{\chi_{11}}$. It turns out that both
signs can occur, depending on the current $I$. For instance, when
$I=7$, we get
$$\beta(7)=0.785,\; \gm_{11}(7)=0.652,\;
c_{1111}(7)=0.375,\;\;\chi_{11}(7)= -0.592,$$ while
$$\beta(11)=0.403,\; \gm_{11}(11)=0.449,\;
c_{1111}(11)=0.924,\;\;\chi_{11}(11)= 0.03.$$ A careful computation
shows that $\chi_{11}$ vanishes at $I\approx 10.93$. We denote this
critical value by $I_{k}$ and conclude that the bifurcation from (N)
to (S) is stable (type II) for $I<I_k$ and unstable for $I_k<I<I_c$.
Finally we point out that the expansion in this section breaks down
for $I$ near $I_k$, and one needs to proceed to higher order terms
there.

\section{The bifurcation from (N) to (P)}\label{NP}

In this section we compute asymptotic approximations for the
solution in the oscillatory (P) state. In this state the solution
$\psi(x,t)$ is time-periodic. We assume in this section that the
current $I$ is fixed in the regime $I>I_{c}$. The transition to the
(P) state takes the form of a Hopf bifurcation; namely, the real
part of the spectrum is zero, and the bottom of the spectrum
consists of a conjugate pair of purely imaginary eigenvalues. We
point out that the spectrum of the operator $M$ (cf. \eqref{gl13})
{\it does} have a nonzero real part, but this real part is exactly
balanced at the transition curve by our choice in this section of
$\Gm=\Gamma_1(I)+\e=\reu+\e$ with $\e>0$ in \eqref{reforma}. Hence,
if we extend our definition \eqref{Lone} of the linear operator
$L_{1}$ to include the case where $\lambda_{1}$ is complex via \beq
L_1 u:=u_{xx}+ixIu+(\reu) u, \eeq then it is $L_{1}$ that possesses
a pair of purely imaginary eigenvalues with corresponding
eigenfunctions $u_{1}$ and $u_{2}$ satisfying \beq L_{1}
u_{1}=-i(\imu) u_{1}\quad\mbox{and}\quad L_{1} u_{2}=i(\imu)
u_{2}.\label{imeig} \eeq We choose to normalize the eigenfunctions
so that $u_j(0)=1,\;j=1,2$, and we assume $\lambda_{1}$ is defined
so that $\imu>0$. With the above choice of $\Gamma$ and definition
of $L_{1}$, we again find that (\ref{gl1})-(\ref{gl3}) takes the
form \eqref{reformb}.

At leading order, we expect the solution to be comprised of a linear
combination of solutions to the equation $\psi_{t}=L_{1} \psi$. This
leads us to seek a periodic solution to \eqref{reformb} of the form
\begin{equation}
\psi(x,t) \sim \eps^{1/2}\left(\al_1(\tau) e^{-i \imu t}u_1(x) +
\al_2(\tau) e^{i \imu t}u_2(x) \right) +\eps^{3/2} \psi_1(x,t) +...
\label{h79}
\end{equation}
Here $\psi_{1}$ is assumed to be a periodic function of $t$ with
period $p_{\e}=2\pi/\imu\,+\;\mathcal{O}(\e)$ and $\al_1, \al_2 $
are coefficients that we expect to evolve slowly in time. We thus
have set $\al_i=\al_i(\tau)$, where, just as before, $\tau=\eps t$.
Our goal here is to compute the functions $\al_i(\tau)$ and thus to
obtain completely the leading order term in the expansion.

Substituting the ansatz (\ref{h79}) into equations \eqref{reformb}
shows that the $O(\eps^{1/2})$ terms are balanced by the choice
above for $\psi$, with the coefficients $\al_i(\tau)$ not yet
determined. Proceeding to the $O(\eps^{3/2})$ level, we obtain \bea
(\psi_{1})_{t}-L_{1}\psi_{1}  =&&  e^{-i \imu t}u_1 \left(\al_1 -
\al_{1 \tau}\right) +  e^{i \imu t} u_2\left(\al_2 - \al_{2
\tau}\right) \\ &&\qquad+ \quad \mathcal{N}[ \al_1 e^{-i \imu t}u_1
+ \al_2 e^{i \imu t}u_2], \label{h81} \eea cf. \eqref{nonnot}.

To obtain a first solvability condition for $\psi_1$, we multiply
equation (\ref{h81}) by $e^{i\imu t}u_1(x)$ and integrate over
$[-1,1]\times[0,2\pi/\imu]$. A second solvability condition is
obtained by integrating similarly against the function $e^{-i\imu
t}u_2(x)$. We note that through \eqref{imeig} and the assumed periodicity of
$\psi_{1}$, the first integration against the left-hand side of
\eqref{h81} yields, after an integration by parts:
\begin{eqnarray*}
&&
\int_{-1}^{1}\int_{0}^{2\pi/\imu}\left((\psi_{1})_{t}-L_{1}\psi_{1}\right)
e^{i\imu t}u_{1}\,dt\,dx\\
&&=-i\imu\int_{-1}^{1}\int_{0}^{2\pi/\imu}\psi_{1}e^{i\imu
t}u_{1}\,dt\,dx-\int_{-1}^{1}\int_{0}^{2\pi/\imu}\psi_{1}e^{i\imu
t}L_{1}u_{1}\,dt\,dx=\mathcal{O}(\e).
\end{eqnarray*}
Similarly, the left-hand side in the
second integration vanishes to leading order.

After a lengthy but straight-forward calculation, these two
integrations then give rise to a pair of equations for the
coefficients $\al_i(\tau)$, namely
\begin{eqnarray}
(\al_{1})_{\tau}=\al_1+\left(\chi_{11} |\al_1|^2 + \chi_{12}
|\al_2|^2\right)\al_1,\label{h85a} \\
(\al_{2})_{\tau}=\al_2+\left(\chi^*_{11} |\al_2|^2 + \chi^*_{12}
|\al_1|^2\right)\al_2.\ \label{h85b}
\end{eqnarray}
Here the coefficient $\chi_{11}$ is again given by \eqref{h287}
while $\chi_{12}$ is defined by through:
\begin{equation}
\chi_{12}=\left(\frac 12 c_{1122}+\frac 12
c_{1212}-2\gm_{21}\right)/\bt, \label{h87}
\end{equation}
where
\begin{equation}
\bt:=\int_{-1}^1 u_1^2\;dx,\;\gm_{ij}= \int_{-1}^1 |u_{i}|^2
u_j^2\;dx,\; c_{ijkl}=\int_{-1}^1 u_{i} u_j \theta_{kl} \;dx, \;
\theta_{kl}=\int_{0}^x \left(u_k (u_l)_x^* - u_l^*
(u_k)_x\right)\;dx'. \label{h89}
\end{equation}
We note that the notation above is consistent with \eqref{h289}.

To analyze the evolution of the $\al_i's$, we note that from
\eqref{h85a}-\eqref{h85b} it is easy to derive the system
\begin{eqnarray}
\left(|\al_{1}|\right)_{\tau}=|\al_1|(1+({\rm{Re}}\,\chi_{11}
|\al_1|^2 + {\rm{Re}}\,
\chi_{12} |\al_2|^2),\label{mh85a} \\
\left(|\al_{2}|\right)_{\tau}=|\al_2|(1+({\rm{Re}}\,\chi_{11}
|\al_2|^2 + {\rm{Re}}\, \chi_{12} |\al_1|^2), \label{mh85b}
\end{eqnarray}
governing the evolution of the moduli of the $\al_i$. Now if the
initial conditions for \eqref{reformb} are taken such that
$\al_{1}(0)\not=0$ and $\al_{2}(0)\not=0$, then it follows from
\eqref{mh85a}-\eqref{mh85b} that $\al_{1}$ and $\al_{2}$ are
non-zero for all future times. In this case, we introduce the
Ricatti transform $r:=|\alpha_1|/|\alpha_2|$. From
\eqref{mh85a}-\eqref{mh85b} we obtain that
\begin{equation}
\label{req}
\begin{aligned}
r'&=\frac{|\alpha_2||\alpha_1|'-|\alpha_2|'|\alpha_1|}{|\alpha_2|^2}\\
&= \frac{|\alpha_1||\alpha_2|({\rm{Re}}\, \chi_{11}-{\rm{Re}}\,
\chi_{12})
(|\alpha_1|^2-|\alpha_2|^2) }{|\alpha_2|^2}\\
&= ({\rm{Re}}\, \chi_{11}-{\rm{Re}}\,\chi_{12}) |\alpha_2|^2 r
(r^2-1)=-({\rm{Re}}\,\hat{\chi})\, |\alpha_2|^2 r (r^2-1),
\end{aligned}
\end{equation}
where $\cdot'$ denotes $\frac{d}{d\tau}$ and we have introduced the
complex constant \beq
\hat{\chi}:=\chi_{12}-\chi_{11}.\label{chihatdefn}\eeq Provided that
${\rm{Re}}\,\hat{\chi}>0,$ we learn from \eqref{req} that $r\to 1$
at an exponential rate as $\tau\to\infty.$ One can indeed check
numerically that for $I>I_c$, the inequality
${\rm{Re}}\,\hat{\chi}>0$ holds. See Figure \ref{chi_1}.

Returning to the system \eqref{mh85a}-\eqref{mh85b} with this
information, we can determine the asymptotic value of the modulus of
both $\al_1$ and $\al_2$ to be
\[\abs{\al_1(\tau)}\sim\abs{\al_2(\tau)}\sim\sqrt{\frac{-1}{({\rm Re} \chi_{11} + {\rm Re}
\chi_{12})}}=\sqrt{\frac{1}{{\rm{Re}}\,\tilde{\chi}}}\quad\mbox{for}\;\tau>>1\]
where we have introduced another complex constant \beq
\tilde{\chi}:= -(\chi_{11}+\chi_{12}).\label{chitildedefn} \eeq
Numerical calculation reveals that ${\rm{Re}}\,\tilde{\chi}>0$ for
$I>I_c$ as well. Again, see Figure \ref{chi_1}.

\begin{figure}
\begin{center}
\begin{tabular}{c}
\includegraphics[height=7cm]{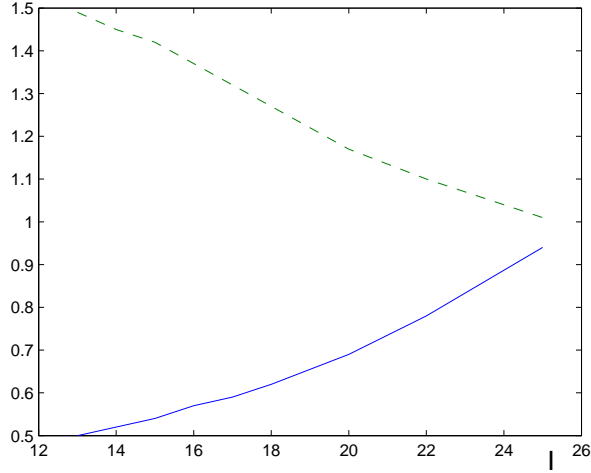}
\end{tabular}
\end{center}
\caption{The dashed curve represents the graph of the parameter
${\rm{Re}}\,\hat{\chi}$ as a function of applied current $I$ and the
solid curve represents the graph of ${\rm{Re}}\,\tilde{\chi}.$ }
\label{chi_1}
\end{figure}

 Substitution into \eqref{h85a}-\eqref{h85b} then yields a linear
dependence on $\tau$ of the phase of both $\al_1$ and $\al_2$ for
$\tau$ large with
\begin{equation}
\al_1(\tau)=\al_2^*(\tau)=
\sqrt{\frac{1}{{\rm{Re}}\,\tilde{\chi}}}\, e^{i \om
\tau}\quad\mbox{for}\;\tau>>1,\label{asymal}\end{equation} where
$\om  := -\frac{{\rm{Im}}\,\tilde{\chi}}{{\rm{Re}}\,\tilde{\chi}}. $
Of course, in light of the rotational invariance of all of the above
equations, the asymptotic state \eqref{asymal} holds only up to a
constant rotation. As an example, we computed the functionals
$\chi_{11}$ and $\chi_{12}$ for $I=20$. In this case the leading
eigenvalues are $\lm_{1,2} \approx 8.64 \pm 5.25\;i$, the amplitude
is $\approx 0.92$ and $\om \approx -1.81$. Unlike the case $I<I_c$,
the transition from (N) to (P) is always stable.

The case where either $\al_{1}(0)=0$ or $\al_{2}(0)=0$ is treated
separately. In view of \eqref{h85a}-\eqref{h85b}, if, for example,
$\al_{1}(0)=0$, then necessarily $\al_{1}(\tau)\equiv 0$. Hence,
$\al_{2}$ evolves by
$$
(\al_{2})_{\tau}=\left(1+\chi^*_{11} |\al_2|^2 \right)\al_2.
$$
From this and \eqref{mh85b} it easily follows that
\begin{equation}
\al_{2}(\tau)\sim\frac{1}{\sqrt{-{\rm{Re}}\,\chi_{11}}}e^{i\frac{{\rm{Im}}\chi_{11}}{{\rm{Re}}\chi_{11}}\,\tau}\quad
\mbox{for}\;\tau>>1\label{a2unst}
\end{equation}
provided that ${\rm{Re}}\,\chi_{11}<0$. Since
$\chi_{11}=-\frac{1}{2}(\tilde{\chi}+\hat{\chi})$, and both
$\tilde{\chi}$ and $\hat{\chi}$ have been shown numerically to be
positive, we do indeed have this condition on $\chi_{11}$ met.
Similarly, when $\al_{2}(0)=0$, one finds that $\al_{2}(\tau)\equiv
0$ and
\begin{equation}\label{a1unst}
\al_{1}(\tau)\sim\frac{1}{\sqrt{-{\rm{Re}}\,\chi_{11}}}e^{-i\frac{{\rm{Im}}\chi_{11}}{{\rm{Re}}\chi_{11}}\,\tau}\quad
\mbox{for}\;\tau>>1
\end{equation}
provided again that ${\rm{Re}}\chi_{11}<0$.

To summarize, in the generic case where neither $\al_{1}(0)$ nor
$\al_{2}(0)$ vanish, the solution in the (P) state is to leading
order
\begin{equation}
\psi(x,t) \sim \eps^{1/2} A \left( e^{-i( \imu +\om \eps)t} u_1(x) +
e^{i( \imu+ \om \eps)t} u_2(x)\right). \label{h95}
\end{equation}
Therefore the solution is time-periodic with period
\begin{equation}
p_{\e}\sim 2\pi/\left(\imu+\om \eps\right)+o(\e). \label{951}
\end{equation}
In a subsequent article, we will examine in some detail how the
asymptotic solution (\ref{h95}) extends deeper into the nonlinear
regime of the (P) state.

\section{Rigorous bifurcation theory}
In this section we will establish a rigorous justification for the
formally derived expansions and bifurcations of the previous two
sections. As was done earlier, we will set the normal conductivity
$\sigma$ equal to one and set the interval $[-L,L]$ to be $[-1,1]$,
thus focusing on the interplay between the externally forced current
$I$ and the temperature-dependent parameter $\Gamma$. As was done
earlier, we will treat the case of Dirichlet boundary conditions,
though similar conclusions can be rigorously established for the
case of Neumann boundary conditions. Though in our previously
derived formal asymptotics, we only pursued the case of bifurcation
off of the principal eigenvalue, at the end of this section we also
treat bifurcation off of any eigenvalue.

\subsection{Spectral properties of linear operator}

We begin by recalling some notation and collecting some facts about
the linear operator $M$ given by \eqref{gl13}.

\begin{lemma}\label{Mspec}
The spectrum of $M$ consists only of point spectrum, denoted by
$\{-\lambda_j\}$ with corresponding eigenfunctions $\{u_j\}.$ If
$(\lambda_j,u_j)$ is an eigenpair satisfying
$$
M u_j=-\lambda_j u_j,\quad u_j(\pm 1)=0,
$$
then \beq {\rm{Re}}\,\lambda_j>0,\quad\mbox{and}\quad
\abs{{\rm{Im}}\,\lambda_j}< I.\label{ReIm}\eeq Thus, in particular
we may order the eigenvalues $\lambda_1,\lambda_2,\ldots$ according
to the size of their real part, with $0<{\rm{Re}}\,\lambda_1\leq
{\rm{Re}}\,\lambda_2\leq \ldots.$ The PT-symmetry of the operator is
reflected in the fact that if $(\lambda_j,u_j)$ is an eigenpair then
so is $(\lambda^*_j,u_j^{\da})$ where $u_j^{\da}(x):=u^*_j(-x).$
Finally, for each positive integer $\ell$, there exists a positive
value of the current $I$, which we denote by $I_{\ell}$, with
$I_{\ell}<I_{\ell+1}$ such that \beq
\lambda_{2\ell-1},\;\lambda_{2\ell}\in\R\quad\mbox{for}\quad I\leq
I_{\ell}\quad\mbox{while}\quad
\lambda_{2\ell-1}=\lambda_{2\ell}^*\,\not\in\R\quad\mbox{for}\quad
I> I_{\ell}.\label{split}\eeq
\end{lemma}
\begin{remark} {\rm We note that in earlier parts of this
paper, the critical value $I_1$ was denoted by $I_c$. In the next
subsection, we revert to this notation to keep consistency with
earlier sections.}
\end{remark}
\begin{proof}
The fact that the spectrum consists entirely of eigenvalues follows
from standard spectral theory. To verify \eqref{ReIm}, multiply the
equation $M u=-\lambda u$ by $u^*$ and integrate to obtain
$$
{\rm{Re}}\,\lambda=\frac{\int_{-1}^{1}\abs{u_x}^2\,dx}{\int_{-1}^{1}\abs{u}^2\,dx}\quad
\mbox{and}\quad{\rm{Im}}\,\lambda=\frac{I\int_{-1}^{1}x\abs{u}^2\,dx}{\int_{-1}^{1}\abs{u}^2\,dx}.
$$
The fact that if $(\lambda,u)$ is an eigenpair then so is
$(\lambda^*,u^{\da})$ follows by inspection. The assertion that the
spectrum is real for small $I$ follows from Theorem 4.1 of
\cite{latr} while the existence of a critical values $\{I_{\ell}\}$
beyond which pairs of real eigenvalues collide to form conjugate
pairs is a result of \cite{shk1,shk2}.
\end{proof}
Next, for any fixed positive integer ${\ell}$, we introduce the
operator \beq L_{\ell}u:=M
u+({\rm{Re}}\,\lambda_{2\ell-1})\,u.\label{Lell} \eeq Then
$L_{\ell}$ has spectrum shifted from that of $M$ by
${\rm{Re}}\,\lambda_{2\ell-1}$ so that $L_{\ell} u_j=-\mu_j\,u_j$
where $\mu_j=\lambda_j-{\rm{Re}}\,\lambda_{2\ell-1}$. Based on the
behavior of the spectrum of $M$ described in Lemma \ref{Mspec}, we
have the following scenario for $L_{\ell}$. \bea
&&\qquad\mbox{For}\;0\leq I<I_{\ell}:\quad\label{Iless}\\
&& \mu_{2\ell-1}=0,\quad {\rm{Re}}\,\mu_j<0\;\mbox{for}\;1\leq
j<2\ell-1,\quad \mu_j>0\;\mbox{for}\;j>2\ell-1.\nonumber\\
&&\nonumber\\
&&\qquad\mbox{For}\;I=I_{\ell}:\label{Iz}\\
&& \mu_{2\ell-1}=\mu_{2\ell}=0,\quad
{\rm{Re}}\,\mu_j<0\;\mbox{for}\;1\leq
j<2\ell-1,\quad \mu_j>0\;\mbox{for}\;j>2\ell.\nonumber\\
&&\nonumber\\
&&\qquad\mbox{For}\;I>I_{\ell}:\label{Imore}\\
&&{\rm{Re}}\,\mu_{2\ell-1}={\rm{Re}}\,\mu_{2\ell}=0,\quad
{\rm{Im}}\,\mu_{2\ell-1}=-{\rm{Im}}\,\mu_{2\ell}\not=0,\nonumber\\
&&{\rm{Re}}\,\mu_j<0\;\mbox{for}\;1\leq j<2\ell-1,\quad
{\rm{Re}}\,\mu_j>0\;\mbox{for}\;j>2\ell.\nonumber \eea

We also note that by PT-symmetry, if we normalize all eigenfunctions
so that $u_{\ell}(0)=1$, then we must have \beq
u_{2\ell-1}=u^{\da}_{2\ell-1}\quad\mbox{when}\quad I\leq
I_{\ell}\quad \mbox{while} \quad
u_{2\ell-1}=u^{\da}_{2\ell}\quad\mbox{when}\quad I>I_{l}.\label{eig}
\eeq For later use, we also introduce the spectral gap
$\zeta_{\ell}$ given by \beq
\zeta_{\ell}:=\inf_j\,\{|{\rm{Re}}\,\mu_j|\,:\;\mu_j\;\mbox{is an
eigenvalue of}\;L_{\ell}\;\mbox{with non-zero real
part}\}.\label{gap}\eeq

In what follows it will be convenient to choose a basis for the
eigenspace of $L_1$ that is PT-symmetric. This has already been
taken care of when $I<I_1$. However, for $I>I_1$ we introduce the
basis $v_1$ and $v_2$ given by \beq v_1:=u_1+u_2\quad\mbox{and}\quad
v_2:=i(u_1-u_2).\label{vdefn}\eeq The PT-symmetry of this basis
follows from \eqref{eig}. Note also that this basis satisfies the
relations \beq L_1 v_1=-\imu v_2,\quad L_1 v_2=\imu
v_1\label{veval}\eeq.

\subsection{Bifurcation from first eigenvalue}
We now develop the rigorous bifurcation theory associated with the
stationary and periodic solution branches formally derived in
earlier sections. To this end, we wish to reformulate the full
nonlinear system \eqref{phi0}-\eqref{gl5} in such a way as to make
it amenable to standard center manifold and bifurcation theory.
Accordingly, we first solve for the electric potential $\phi$ in
\eqref{gl3}, thereby obtaining \beq \phi=-Ix+\frac
i2\int_0^x\left(\psi\psi^*_x-\psi^*\psi_x\right)\,dx'.\label{solvephi1}
\eeq In this section, we will focus on the stable bifurcations that
occurs off of the first eigenvalue of the linear operator and so we
henceforth fix the positive integer $\ell$ from the previous section
to equal $1$ and pick $\Gamma$ to be of the form
$\Gamma={\rm{Re}}\,\lambda_{1}+\e$. Substituting \eqref{solvephi}
into \eqref{gl1}, we obtain a single, nonlocal complex equation so
that \eqref{gl1}-\eqref{gl3} can be rewritten as \beq
\psi_t=L_{1}\psi+\mathcal{N}(\psi,\e)\label{reform}\eeq where we
have introduced $\mathcal{N}(y,\e):=\mathcal{N}[y]+\e y$ with \beq
\mathcal{N}[y]:=-\abs{y}^2 y+ \frac 12\; y\int_0^x\left(y y^*_x-y^*
y_x\right)\,dx'.\label{nonnot1}\eeq We recall that as before, the
system is augmented with Dirichlet boundary conditions $\psi(\pm
1,t)=0$ along with the normalization $\phi(0,t)=0$ for $t\geq 0$,
and initial conditions, say $\psi(x,0)=\psi_0(x)$. We take $\e$ to
be small and positive (unless otherwise specified).

We remark that since $L_{1}y^{\da}=\left(L_{1}y\right)^{\da}$ and
$\mathcal{N}[y^{\da}]=\left(\mathcal{N}[y]\right)^{\da}$, it follows
easily that the flow \eqref{reform} preserves PT-symmetry in the
sense that if $\psi$ is a solution to \eqref{reform}, then so is
$\psi^{\da}.$ Hence, by uniqueness, we note that if the initial data
$\psi_0$ is PT-symmetric, i.e. if $\psi_0=\psi_0^{\da}$, then so is
the resulting solution $\psi.$

For later use, we also record the estimate: \begin{lemma}
\label{cubiclem} There exists a positive constant $C_0$ such that
\beq \norm{\mathcal{N}[y]}_{H^1}\leq C_0 \norm{y}^3_{H^1}\quad
\mbox{for all}\; y\in H^1_0((-1,1));\C).\label{cubic}\eeq\end{lemma}

\begin{proof}
For $y\in H^1_0((-1,1))$ recall that both $\norm{y}_{L^{\infty}}$
and $\norm{y}_{H^1}$ are controlled by $\norm{y_x}_{L^2}.$ We begin
by estimating the $H^1$ norm of the local part of $\mathcal{N}.$ We
find \[ \norm{ |y|^2 y}_{L^2}\leq \norm{y}_{L^6}^3\leq
\norm{y}_{L^{\infty}}^2\norm{y}_{L^2}\leq C \norm{y}^3_{H^1}. \]
Similarly, \[ \norm{ \big(|y|^2 y\big)_x}_{L^2}\leq C
\big(\int_{-1}^{1}\abs{y}^4\abs{y_x}^2\,dx\big)^{1/2}\leq C
\norm{y}_{L^{\infty}}^2\norm{y_x}_{L^2}\leq C \norm{y}^3_{H^1}\]
Turning now to the nonlocal part of $\mathcal{N}$ we find \bea &&
\norm{y\int_{0}^x \left(yy_x^*-y^*y_x)\right)\,dx'}_{L^2}\leq
C\bigg(\int_{-1}^{1}\big(\abs{y}\int_{-1}^{1}\abs{y}\abs{y_{x}}\,dx'\big)^{2}\,dx\bigg)^{1/2}
\nonumber\\
&&\leq
C\big(\int_{-1}^{1}\abs{y}\abs{y_{x}}\,dx\big)\norm{y}_{L^{2}}\leq C
\norm{y}_{L^{2}}^{2}\norm{y_{x}}_{L^{2}}\leq C
\norm{y}^3_{H^1}.\nonumber \eea Finally, we check that \bea &&
\norm{\bigg(y\int_{0}^x
\left(yy_x^*-y^*y_x)\right)\,dx'\bigg)_{x}}_{L^2}\leq\nonumber\\
&&C
\bigg(\int_{-1}^{1}\big(\abs{y_{x}}\int_{-1}^{1}\abs{y}\abs{y_{x}}\,dx'\big)^{2}\,dx\bigg)^{1/2}
+C\bigg(\int_{-1}^{1}\abs{y}^{4}\abs{y_{x}}^{2}\bigg)^{1/2}\nonumber\\
&&\leq C \norm{y}_{L^{2}}\norm{y_{x}}^{2}_{L^{2}}+C
\norm{y}^{2}_{L^{\infty}} \norm{y_{x}}_{L^{2}}\leq C
\norm{y}^3_{H^1}\nonumber \eea
\end{proof}

Armed with a full understanding of the linearized operator and
control on the nonlinear operator provided by Lemma \ref{cubiclem},
we can proceed to construct a center manifold for the flow
\eqref{reform}. To this end, we will denote by $S_c$ the center
subspace associated with $L_{1}$; that is, $S_c$ is the eigenspace
associated with any eigenvalues of $L_1$ having zero real part. For
$I\leq I_{c}\;(=I_1)$, we have $S_c={\rm{span}}\,\{u_{1}\}$ while
for $I>I_{c}$,
$S_c={\rm{span}}\,\{v_{1},v_{2}\}\;(={\rm{span}}\,\{u_{1},u_{2}\})$.

As $\e$ will play the role of a bifurcation parameter, we then
augment \eqref{reform} with the equation \beq
\e_t=0.\label{ezero}\eeq
The theorem below provides for the existence of a finite dimensional invariant manifold associated with
the flow \eqref{reform} for each fixed small $\e$ describing all orbits of sufficiently small norm. This, in effect,
allows us to rigorize the formal bifurcation calculations of the previous two chapters by reducing the analysis of the
nonlocal P.D.E. to a study of a local system of O.D.E.'s with accompanying rigorous error estimates.

\begin{theorem}\label{cm} For each value $I>0$ and positive integer $k$,
there is a $C^k$ local center manifold $\mathcal{M}\subset
H^1((-1,1);\C)\times\R$ of \eqref{reform}, \eqref{ezero} tangent to
the center subspace. The center manifold $\mathcal{M}$ is
expressible as a graph over the center subspace in the sense that
there exists a $C^k$ map $\Phi:S_c\times\R\to H^1((-1,1);\C)$ such
that \beq \mathcal{M}=\bigcup_
{\abs{\e}<\e_0}\,\big(\mathcal{M}_{\e}\times\{\e\}\big)\quad\mbox{where}\quad
\mathcal{M}_{\e}:=\{\Phi(u,\e):\,u\in
S_c,\;\norm{u}_{H^1}<\delta_0,\;\abs{\e}<\e_0\} \label{cmg}\eeq for
some sufficiently small positive constants $\delta_0$ and $\e_0$
depending in particular on $k$.
 The center manifold is invariant under complex rotation,
i.e. $(\psi,\e)\in\mathcal{M}\implies\;(e^{i\theta}\psi,\e)\in
\mathcal{M}$ for all $\theta\in\R$ and in fact \beq
e^{i\theta}\Phi(u,\e)=\Phi(e^{i\theta}u,\e)\quad\mbox{for all}\;
\theta\in\R.\label{phirot}\eeq The center manifold is also
PT-symmetric, i.e.
$(\psi,\e)\in\mathcal{M}\implies\;(\psi^{\da},\e)\in \mathcal{M}$.
If $u=u^{\da}$ then $\Phi(u,\e)=\Phi^{\da}(u,\e).$

The discrepancy between the center manifold and the center subspace can be expressed
through the estimate \beq \norm{\Phi(u,\e)-u}_{H^1}
\leq
C_1\left(\norm{u}_{H^1}^3+\abs{\e}\norm{u}_{H^1}\right)\label{verytang}
\eeq which holds for any
pair $(u,\e)$ such that $u\in S_c$ with $\norm{u}_{H^1}<\delta_0$ and
$\abs{\e}<\e_0$,
where $C_1$ is a positive constant independent of $u$ and $\e$.

The center manifold is locally invariant for the flow \eqref{reform}
in the sense that if $\abs{\e}<\e_0$ and the initial data $\psi_0$
lies on $\mathcal{M}_\e$, then so does the solution $\psi^{\e}$ to
\eqref{reform} so long as $\norm{\psi^\e(\cdot,t)}_{H^1}$ stays
sufficiently small. Hence, for such initial data, one can describe
the resulting solution $\psi^{\e}(t)=\psi^{\e}(\cdot,t)$ through
either one or two maps $\beta^{\e}_j:[0,\infty)\to\C$ via
$\psi^{\e}(t)=\Phi(\beta^{\e}_1(t)u_1,\e)$ when $I<I_c$ or
$\psi^{\e}(t)=\Phi(\beta^{\e}_1(t)v_1,\beta^{\e}_2(t)v_2,\e)$ when
$I>I_c).$ Finally, $\mathcal{M}$ contains all nearby bounded
solutions of \eqref{reform} in $H^{1}$, and in particular, it
contains any nearby steady-state or time-periodic solutions.
\end{theorem}
\begin{proof}
We follow a standard center manifold construction, along the lines
for example, of \cite{Bressan}. To outline this approach, we first
note that in light of conditions \eqref{ReIm}, the spectrum of the
operator $-L_1$ lies within the set
$$
\{\lambda\in \C:\,\abs{{\rm{arg}}\,(\lambda+a)}<\frac{\pi}{4}\},
$$
for some positive number $a=a(I)$. Hence, $-L_1$ is sectorial and we
may assert the existence of an analytic semi-group $\{e^{L_1
t}\}_{t\geq 0}$, cf. \cite{Henry}, Theorem 1.3.4 or \cite{W},
section 2.2.3.

We will denote the ($L^2$) projection operators from $H^1((-1,1);\C)$ onto
the center and stable subspaces of $L_{1}$ by $\Pi_c$ and $\Pi_s$
respectively. Here by stable subspace we mean the span of all eigenvectors
of $L_1$ whose corresonding eigenvalues have negative real aprts.
We note that since the real part of all eigenvalues of
$L_1$ are non-positive, $L_1$ has no unstable subspace.

 A local center manifold
is constructed by first constructing a global center manifold for a
problem with a truncated nonlinearity through the introduction of a
cut-off function $\rho\in C^{\infty}([0,\infty);[0,1])$ satisfying
$\rho(s)\equiv 1$ for $0\leq s\leq 1$ and $\rho(s)\equiv 0$ for
$s\geq 2$. For any $\delta>0$ we then let
$\rho^{\delta}(s):=\rho(s/\delta).$ We use this cut-off to truncate
the nonlinearity $\mathcal{N}(y,\e)$ by defining
$\mathcal{N}^{\delta}(y,\e):=\rho^{\delta}(\norm{y}_{H^1})\,\mathcal{N}(y,\e).$

The graph map $\Phi:S_c\times\R\to H^1((-1,1);\C)$ is then defined by the following procedure:
For $u\in S_c$ and $\e$ fixed, we use a ``variation of constants" approach, rephrasing the P.D.E. \eqref{reform} (with the original nonlinearity
replaced by $\mathcal{N}^{\delta}(y,\e)$) as an
integral equation:
\bea
&&y(x,t)=\Gamma(u,\e,y) :=e^{L_{1}t}u+\int_0^t
e^{L_{1}(t-\tau)}\Pi_c\,\mathcal{N}^{\delta}\left(y(x,\tau),\e\right)\,d\tau\nonumber\\
&&+ \int_{-\infty}^t
e^{L_{1}(t-\tau)}\Pi_s\,\mathcal{N}^{\delta}\left(y(x,\tau),\e\right)\,d\tau.
\nonumber\\
\label{fixed} \eea
Then one argues that there exists a unique fixed point $y_{u,\e}=y_{u,\e}(x,t)$
to \eqref{fixed} in the space of functions
that grow sufficiently slowly at $t=\pm\infty$ given by
$$
Y_{\eta}:=\{y\in
C\left(-\infty,\infty);H^1((-1,1))\right):\;\norm{y}_{\eta}<\infty\}.
$$
Here $\norm{y}_{\eta}:=\sup_{t\in\R}e^{-\eta
|t|}\norm{y(\cdot,t)}_{H^1((-1,1))}$ and $\eta$ is any fixed
positive number less than the spectral gap $\zeta_{1}$, cf.
\eqref{gap}. Once the existence of this fixed point is established, we define the map $\Phi$ by
\[\Phi(u,\e)=y_{u,\e}(\cdot,0).\]

We should remark that when $t$ is negative, we
interpret $e^{L_{1}t}\Pi_c$ in \eqref{fixed} to mean flow projected
onto the finite dimensional center subspace; thus it reduces to a
finite number of ordinary differential equations. To see this and to
carry out the application of the contraction mapping principle to \eqref{fixed},
one considers the inverse Laplace
transform representations \beq
e^{tL_1}\Pi_c:=\int_{\Gamma_c}e^{\lambda t}(\lambda
I-L_1)^{-1}\,d\lambda, \label{bd1}\eeq  and \beq
e^{tL_1}\Pi_s:=\int_{\Gamma_s}e^{\lambda t}(\lambda
I-L_1)^{-1}\,d\lambda,\label{bd3} \eeq where $\Gamma_c$ is a bounded
contour enclosing the eigenvalues with zero real part and $\Gamma_s$ is a contour
in the left half-plane enclosing the stable spectrum that tends
asymptotically to infinity along the lines $a(1-s)\pm a i s$ as
$s\to\pm\infty$.  Since the resolvent is bounded along these
contours, it follows from these representations that for
every $s_1\in (0,\zeta_{1})$ one has \bea
&&\norm{e^{tL_1}\Pi_c}_{H^1\to H^1}\leq C,\quad\mbox{for
all}\;t\in\R,\label{bd4}\\
&& \norm{e^{tL_1}\Pi_s}_{H^1\to H^1}\leq Ce^{-(\zeta_{1}-s_1)t}
\quad\mbox{for all}\;t\geq 0.\label{bd5} \eea

Invoking these bounds, and by choosing the parameter $\delta$ in the
cut-off of the nonlinearity sufficiently small, the existence of a fixed point to
\eqref{fixed} follows from \eqref{cubic} by the contraction
%CHANGED: added this, important for completeness/correctness! -K
%mapping principle.
mapping principle, from which we find easily also Lipshitz regularity
of $\Phi$.
The asserted $C^\infty$ regularity of $\Phi$ may be established
by a careful iterative argument as described in \cite{Bressan}, using
$C^\infty$ regularity of the truncated equations
(in general, the center manifold inherits one degree less regularity
than the underlying equations);
we omit discussion of this delicate point.
%ENDCHANGED
The local center manifold for the untruncated
problem is then realized through \eqref{cmg} by choosing $\delta_0$
and $\e_0$ sufficiently small.

The rotational and PT invariance of $\mathcal{M}$ follow from the
fact that the center subspace $S_c$ enjoys these invariances and the
fact that for any $\theta_0\in\R$ one has
\[
\Gamma(e^{i\theta_0}
u,\e,e^{i\theta_0}y)=e^{i\theta_0}\Gamma(u,\e,y),\quad\mbox{as well
as}\quad
\Gamma(u^{\da},\e,y^{\da})=\left(\Gamma(u,\e,y)\right)^{\da}.
\]
With regard to this last assertion, note in particular that
$\mathcal{N}^{\delta}(y^{\da},\e)=\left(\mathcal{N}^{\delta}(y,\e)\right)^{\da}.$
Also if $u=u^{\da}$ and if $y_{u,\e}=\Gamma(u,\e,y_{u,\e})$ then
necessarily
\[
y^{\da}_{u,\e}=\left(\Gamma(u,\e,y_{u,\e})\right)^{\da}=\Gamma(u,\e,y^{\da}_{u,\e})
\]
and so by the uniqueness of the fixed point, necessarily
$y^{\da}_{u,\e}=y_{u,\e}.$ Hence, in particular
$y^{\da}_{u,\e}(\cdot,0)=y_{u,\e}(\cdot,0)$ and we have
$\Phi(u,\e)^{\da}=\Phi(u,\e).$ Similarly,
$e^{i\theta_0}\Phi(u,\e)=\Phi(e^{i\theta_0}u,\e)$ for all
$\theta_0\in\R.$

Finally, we turn to the verification of \eqref{verytang}. This comes
from an examination of the iteration procedure leading to the fixed
point as follows. Picking $\delta_0$ sufficiently small, we may
argue that, for instance,
\[
\norm{\Gamma(u,\e,y_1)-\Gamma(u,\e,y_2)}_{\eta}<\frac12\norm{y_1-y_2}_{\eta}
\]
for all $u\in S_c$, all sufficiently small $\e$ and all
$y_1,\,y_2\in Y_{\eta}.$ Letting $y_u$ denote the solution to the
linear problem, i.e. $y_u:=e^{L_{1}t}u$, it then easily follows that
\[ \norm{y_u-y_{u,\e}}_{\eta}=
\norm{y_u-\Gamma(u,\e,y_{u,\e})}_{\eta}\leq
\norm{y_u-\Gamma(u,\e,y_{u})}_{\eta}+\norm{\Gamma(u,\e,y_{u})-\Gamma(u,\e,y_{u,\e})}_{\eta}
\] and so \[ \norm{y_u-y_{u,\e}}_{\eta}\leq 2\norm{y_u-\Gamma(u,\e,y_u)}_{\eta}.\]

Then we calculate \bea
&&\norm{\Phi(u,\e)-u}_{H^1}=\norm{y_{u,\e}(\cdot,0)-y_{u}(\cdot,0)}_{H^1}\nonumber\\
&&\leq \sup_{t\in\R}
\norm{y_{u,\e}(\cdot,t)-y_{u}(\cdot,t)}_{H^1}e^{-\eta|t|}
=\norm{y_{u,\e}-y_{u}}_{\eta}\nonumber\\
&&\leq 2\norm{y_u-\Gamma(u,\e,y_u)}_{\eta}\leq 2C
\left(\norm{u}_{H^1}^3+\abs{\e}\norm{u}_{H^1}\right),\nonumber \eea
where in the last estimate we invoked \eqref{cubic}, \eqref{fixed}
and \eqref{bd1}--\eqref{bd3}.

\end{proof}
We also will need a version of the standard result on exponential
attraction to an orbit on the center manifold in the absence of any
unstable manifold. Again the proof we sketch is an adaptation of a
more general but somewhat weaker result
in \cite{Bressan} that is valid in the presence of an unstable manifold.

\begin{theorem}\label{exp}
For any positive integer $k$ and $r=r(k)>0$ sufficiently small,
there exists a
%CHANGED (need this later)-K:
%Lipshitz
$C^k$
%ENDCHANGED
map $P_{\e}$
from $B(0,r)\subset H^1$ to $\mathcal{M}_{\e}$, equal to the
identity when restricted to $\mathcal{M}_{\e}$, such that, for all
solutions $\psi^{\e}$ of \eqref{reform} originating at time $t=0$
within $B(0,r)$, \beq \label{expbd} \|\psi^{\e}(t)-\hat
\psi^{\e}(t)\|_{H^1}\le C_1e^{-\eta t}d_{H^1}(\psi^{\e}(0),
\mathcal{M}_{\e}), \eeq so long as $\psi^{\e}$ remains in $B(0,r)$,
where $\hat \psi^{\e} \in \mathcal{M}_{\e}$ denotes the trajectory
along $\mathcal{M}_{\e}$ originating at time $t=0$ at
%CHANGED (notation):
%$P_{\e}\psi^{\e}(0)$ and $\eta>0$ and $C_1$ are uniform constants.
$P_{\e}(\psi^{\e}(0))$ and $\eta>0$ and $C_1$ are uniform constants.
%ENDCHANGED
Here $d_{H^1}(\cdot,\cdot)$ denotes the distance in $H^1.$
\end{theorem}

\begin{proof}
As noted in \cite{TZ}, this follows by the proof of the more general
approximation property (v) of the Center Manifold Theorem stated in
\cite{Bressan}, restricted to the case that the underlying
linearized operator ($L_1$ in this case) has no unstable manifold.

To say a bit more about our adaptation of the approach presented in
\cite{Bressan}, given a solution $\psi^{\e}=\psi^{\e}(x,t)$ to
\eqref{reform}, one first extends $\psi^{\e}$ to a function
$\bar{\psi^{\e}}$ defined for negative $t$-values via
\[\bar{\psi^{\e}}=\left\{\begin{matrix} \psi^{\e}
& \quad\mbox{for}\;t\geq 0\\
\psi^{\e}(0)& \quad\mbox{for}\;t< 0.\end{matrix}\right.
\]
Thus, $\bar{\psi^{\e}}$ represents a globally bounded solution to
the equation $\bar{\psi^{\e}}_t=L_1
\bar{\psi^{\e}}+\mathcal{N}(\bar{\psi^{\e}},\e)+\phi^{\e}$ where
\beq\phi^{\e}=\phi^{\e}(x,t):=\left\{\begin{matrix} 0
& \quad\mbox{for}\;t> 0\\
-L_1 \psi^{\e}(0)-\mathcal{N}(\psi^{\e},\e)& \quad\mbox{for}\;t<
0.\end{matrix}\right. \label{phiL}\eeq Then we fix any positive
$\eta$ such that $\eta<\zeta_1$ (cf. \eqref{gap}) and seek a
function $z$ such that $\bar{\psi^{\e}}+z\in \mathcal{M}_{\e}.$ We
will find such a $z$ in the set
\[ Z_{\eta}:= \{z\in
C\left((\infty,\infty);H^1((-1,1))\right):\;|z|_{\eta}<\infty\}
\]
where $|z|_{\eta}:=\sup_{t\in\R}e^{\eta
t}\norm{z(\cdot,t)}_{H^1((-1,1))}$ and then define the projection
$P_{\e}$ onto $\mathcal{M}_{\e}$ via
$P_{\e}(\psi^{\e}(0)):=\psi^{\e}(0)+z(0).$ Thus, the trajectory on
the center manifold satisfying \eqref{expbd} will be
$\hat{\psi^{\e}}(t):=\psi^{\e}(t)+z(t).$ The function $z$ is
produced as follows: plugging $\bar{\psi^{\e}}+z$ into the integral
equation \eqref{fixed} leads one to seek $z$ as a fixed point of the
mapping $\Lambda:Z_{\eta}\to Z_{\eta}$ defined by \bea &&
\Lambda(z):=
-\int_t^{\infty}e^{L_{1}(t-\tau)}\Pi_c\,\left[\mathcal{N}^{\delta}\left(\bar{\psi^{\e}}(x,\tau)+z(x,\tau),\e\right)
-\mathcal{N}^{\delta}\left(\bar{\psi^{\e}}(x,\tau),\e\right)\right]\,d\tau+\nonumber\\
&&\int_t^{\infty}e^{L_{1}(t-\tau)}\Pi_c\,\phi^{\e}(x,\tau)\,d\tau +
\int_{-\infty}^t e^{L_{1}(t-\tau)}\Pi_s\,
\left[\mathcal{N}^{\delta}\left(\bar{\psi^{\e}}(x,\tau)+z(x,\tau),\e\right)
-\mathcal{N}^{\delta}\left(\bar{\psi^{\e}}(x,\tau),\e\right)\right]\nonumber\\
&&-\int_{-\infty}^t e^{L_{1}(t-\tau)}\Pi_s\,\phi^{\e}(x,\tau)
\,d\tau.
\nonumber\\
\label{newfixed} \eea Again the existence of a (unique such) fixed point
$z\in Z_{\eta}$ follows readily from the contraction mapping
principle since one can check that
$|\Lambda(z_1)-\Lambda(z_2)|_{\eta}\leq \theta |z_1-z_2|_{\eta}$ for
some $\theta\in (0,1).$
%CHANGED:
The fixed point $z$, and thus the map $P_\eps$ that it determines,
is Lipshitz in $\psi^\eps(0)$ by construction. With further effort,
it may be shown to be $C^k$ for any $k$, by a procedure similar to
that used to show smoothness of the center manifold \cite{Bressan}
in the analogous fixed-point construction of Proposition \ref{cm},
using $C^k$ regularity of both the center manifold and the truncated
equations. (In general, $P_\eps$ inherits the regularity of the
center manifold.) Note that the center manifold solution
$\hat{\psi^{\e}}(t)$ so constructed satisfies the truncated
equations \eqref{fixed} and not \eqref{reform}, since the righthand
side of \eqref{newfixed} involves the truncated nonlinearity
$\mathcal{N}^\delta$ in place of $\mathcal{N}$. However, this makes
no difference since the equations agree on the ball $B(0,r)$ under
consideration, for $r>0$ sufficiently small.

Hence, we have \[  |z|_{\eta}\leq
|z-\Lambda(0)|_{\eta}+|\Lambda(0)|_{\eta}
=|\Lambda(z)-\Lambda(0)|_{\eta}+|\Lambda(0)|_{\eta}\leq \theta
|z|_{\eta}+|\Lambda(0)|_{\eta},\] and so we conclude that \beq
|z|_{\eta}\le C|\Lambda(0)|_{\eta}. \label{zeieq}\eeq One then
observes from \eqref{phiL} and \eqref{newfixed} that for $t\geq 0$,
\beq \Lambda(0)(t)= e^{L_1 t}\Pi_s F(\psi^{\e}(0)),\label{Flamb}\eeq
where $F:H^1\to H^1$ is given by \[F(v):=\int_{-\infty}^0 e^{-L_1
\tau} \Pi_s \left(L_1 v+ N^\delta(v,\e)\right)\,d\tau.\] Note that
$F$ is evidently bounded and Lipshitz.
We next claim that if
$\psi(0)\in \mathcal{M}_{\e}$, then $F(\psi^{\e}(0))=0.$ To see
this, note that in this case the invariance property of the center
manifold implies that $\psi^{\e}\in\mathcal{M}_{\e}$ for $t\not =0$
as well. Hence the unique fixed point of \eqref{newfixed} must be
$z\equiv 0$ and so in particular $\Lambda(0)(t)=0$ for $t>0$, which
establishes the claim.

Finally, fixing any element $\psi^{\e}(0)\in
H^1\setminus\mathcal{M}_{\e}$ with sufficiently small $H^1$-norm,
one uses this last observation to obtain
\beq\norm{F(\psi^{\e}(0))}_{H^1}=\inf_{\psi_1\in\mathcal{M}{\e}}\norm{F(\psi^{\e}(0))-F(\psi_1)}_{H^1}\leq
Cd_{H^1}(\psi^{\e}(0),\mathcal{M}_{\e}). \label{Fin}\eeq The bound
\eqref{expbd} now follows by combining \eqref{zeieq}, \eqref{Flamb}
and \eqref{Fin} and using the bound \eqref{bd5}, since
$z=\hat{\psi}-\psi.$
\end{proof}

%CHANGED (added for later reference):
\begin{remark}\label{foliation}
{\rm Using the Implicit Function Theorem and the fact that $P_\eps$
is the identity on $\cMe$, we find that $H^1$ is foliated on a small
neighborhood of $\cMe$ by transverse smooth manifolds
$P_\eps^{-1}(w)$ through each $w\in \cMe$, depending in a smooth
fashion on the value of $w$. In particular, for $\overline{C}>0$
sufficiently large and $a>0$ sufficiently small, the $H^1$-ball
$B(0,a)$ is foliated by leaves $P_\eps^{-1}(w)$ for $w\in \cMe \cap
B(0,\overline{C}a)$, carried one to the other under the flow of the
underlying ODE, uniquely specified by the property that each
solution initiating in $P_\eps^{-1}(w_0)$ approaches the solution on
the center manifold with initial data $w_0$ at uniform exponential
rate $\sim 1>>\eps$. As a consequence, a $C^k$ stable manifold
$\cN\subset\cMe\cap B(0,a)$ of an orbit or manifold of orbits within
the center manifold $\cMe$ extends to a $C^k$ stable manifold
$\tilde N:=\cup_{w\in \cN} P_\eps^{-1}(w)\cap B(0,a)$ in $B(0,a)$ of
the same codimension in $B(0,a)$ as the codimension of $\cN$ in
$\cMe$. That is, not only is asymptotic stability in $B(0,a)$
determined completely by asymptotic stability within the center
manifold, but also
%the degree of
conditional stability as measured by codimension of the stable
manifold. }
%ENDCHANGED
\end{remark}

%NOTE FOR KZ's ODE CLASS: similar foliations work in presence
%of unstable manifold, now into stable AND unstable directions,
%following more general result of Bressan.  What we lose is the property of
%uniform exponential approach, NOT the foliation.
%TODO: tricky point is to establish correct dim./codimension of these objects...
%obvious in this case, less obvious maybe for general case... -KZ

We now apply the previous result on existence of a center manifold
to assert the existence of bifurcating stationary and periodic
states for equation \eqref{reform}.

We begin with the case of stationary states bifurcating from the
normal state. We refer to Section 4 for the definition \eqref{h287}
of the parameter $\chi_{11}$ which was found numerically to be real
for $I\leq I_c$, positive for $I_k<I<I_c$ and negative for $0<I<I_k$
where $I_k\approx 10.93$ and $I_c\approx 12.31.$
\begin{proposition}\label{stat}
For $I$ fixed in the interval $(0,I_k)$, equation \eqref{reform}
exhibits a stable supercritical pitchfork bifurcation of stationary
states $\{e^{i\theta_0} \psi_e(\cdot,\e):\;\theta_0\in [0,2\pi)\}$
branching from the normal state for all sufficiently small and
positive values of $\e$.  These equilibria satisfy the bound
 \beq
 \norm{\psi_e-\,(\frac{1}{\sqrt{-\chi_{11}}})\,\e^{1/2}\,u_1}_{H^1}<C\e^{3/2}
 \label{stable1}\eeq
as predicted formally in Section 4

For $I$ fixed in the interval $(I_k,I_c)$, the equation exhibits an
unstable subcritical pitchfork bifurcation of stationary states
$e^{i\theta_0}\tilde{\psi}_e(\cdot,\e),\;\theta_0\in [0,2\pi),$
branching from the normal state for all sufficiently small and
negative values of $\e$. These equilibria satisfy the bound
 \beq
 \norm{\tilde{\psi}_e-\,(\frac{1}{\sqrt{\chi_{11}}})\,|\e|^{1/2}\,u_1}_{H^1}<C|\e|^{3/2}
 \label{unstable1}\eeq
as predicted formally in Section 4.
\end{proposition}

\begin{proof} Since we work here in the setting where $I\leq I_c$,
the center subspace is spanned by the single eigenfunction $u_1.$
Hence, we will express any point on the center manifold
$\mathcal{M}$ as $(\Phi(\beta,\e),\e)$ where $\beta\in\C$
corresponds to the coefficient of the point $\beta u_1$ on the
center subspace. We begin with the case where $I$ is fixed to lie in
the interval $(0,I_k)$. Then for any small value $\beta_0\in\C$ we
let $\psi^{\e}$ denote the solution to \eqref{reform} satisfying
Dirichlet boundary conditions and initial condition
$\psi^{\e}(\cdot,0)=\Phi(\beta_0,\e)$. For all small, positive $t$,
we know from Theorem \ref{cm} that
$(\psi^{\e}(\cdot,t),\e)\in\mathcal{M}$ and so there exists a smooth
function, which we denote by $\beta^{\e}=\beta^{\e}(t)$, such that
\beq \psi^{\e}(x,t)=\Phi(\beta^{\e}(t),\e). \label{eqa}\eeq

Recalling the definition $\beta_1:=\int_{-1}^{1}u_1^2\,dx$ we now
apply the projection $\Pi_c$ to every term in the equation
\[
\psi^{\e}_t=L_{1}\psi^{\e}+\mathcal{N}(\psi^{\e},\e)
\]
satisfied by $\psi^{\e}$ and then integrate against $u_1.$ We find
that
\[
\int_{-1}^{1}\Pi_c\big(\psi^{\e}_t\big)u_1\,dx=\frac{\partial}{\partial
t}\int_{-1}^{1}\Pi_c\big(\psi^{\e}\big)u_1\,dx=\frac{\partial}{\partial
t}\int_{-1}^{1}\Pi_c\big(\Phi(\beta^{\e},\e)u_1\,dx\big)=\beta_1\beta^{\e}_t.
\]
We also have
\[
\int_{-1}^{1}\Pi_c\big(L_1
\psi^{\e}\big)u_1\,dx=\beta_1\int_{-1}^{1}u_1\,L_1 \psi^{\e}\,dx
=\beta_1\int_{-1}^{1}L_1 u_1\,\psi^{\e}\,dx=0
\]
and $
\int_{-1}^1\Pi_c\big(\e\psi^{\e}\big)u_1\,dx=\e\beta_1\beta^{\e}$.
Consequently, we obtain for $\beta^{\e}$ the O.D.E. \bea
\beta^{\e}_t &=&\e\beta^{\e}+\frac{1}{\beta_1}
\int_{-1}^1 \Pi_c\big(\mathcal{N}[\beta^{\e}u_1]\big)u_1\,dx+e(\beta^{\e},\e)\nonumber\\
&=&\e\beta^{\e}+\chi_{11}\abs{\beta^{\e}}^2\beta^{\e}+e(\beta^{\e},\e)\label{stateode}
\eea where we recall the calculation of
$\int_{-1}^1\Pi_c\big(\mathcal{N}[\beta^{\e}u_1]\big)u_1\,dx$
carried out in Section 4, and we have introduced \beq
e(\beta^{\e},\e):=\frac{1}{\beta_1}\Pi_c\big(\mathcal{N}[\Phi(\beta^{\e},\e)]\big)-\Pi_c\big(\mathcal{N}[
\beta^{\e}u_1]\big).
 \label{edefn}\eeq
Due to \eqref{phirot}, we know that \beq
e(e^{i\theta_0}\beta^{\e},\e)=e^{i\theta_0}e(\beta^{\e},\e)\quad\mbox{for
any}\;\theta_0\in\R\label{erot}\eeq and through elementary use of
the triangle and Cauchy-Schwartz inequalities applied to the
nonlinearity $\mathcal{N}$, along with \eqref{verytang}, we estimate
\bea && \abs{e(\beta^\e,\e)}\leq C
\norm{\mathcal{N}[\Phi(\beta^{\e},\e)]-\mathcal{N}[\beta^{\e}u_1]}_{L^{\infty}} \nonumber\\
&&\leq C\left(
\norm{\Phi(\beta^{\e},\e)}^2_{H^1}+\norm{\beta^{\e}u_1}^2_{H^1}\right)
\left(\norm{\Phi(\beta^{\e},\e)-\beta^{\e}u_1}_{H^1}\right)\nonumber\\
&&=\mathcal{O}\left(\e(\beta^{\e})^3+(\beta^{\e})^5\right).
\label{erest} \eea

Returning to \eqref{stateode}, consider first the case where
$\beta^{\e}(0)=\beta_0\in\R.$ We first claim that the function
$\beta^{\e}(t)$ must be real. To see this, we begin by noting that
since $\beta_0$ is real, the quantity $\beta_0 u_1$ is PT-symmetric.
Hence, in particular
$\psi^{\e}(\cdot,t_1)=\left(\psi^{\e}(\cdot,t_1)\right)^{\da}$ as
well for any fixed $t_1>0$ since $\psi^{\e}$ satisfies a
PT-symmetric initial condition $\Phi(\beta_0 u_1,\e)$ that is . Now
denote by $\psi^{\e,1}$ the unique solution in $Y_{\eta}$ to the
equation \beq
\psi^{\e,1}=\Gamma(\beta^{\e}(t_1)u_1,\e,\psi^{\e,1}).\label{ab}
\eeq By \eqref{eqa} we have
$\psi^{\e}(\cdot,t_1)=\psi^{\e,1}(\cdot,0),$ and consequently,
$\psi^{\e,1}(\cdot,0)=\left(\psi^{\e,1}(\cdot,0)\right)^{\da}.$
Evaluating \eqref{ab} at $t=0$ and applying the $\da$ operation to
both sides, we then conclude that \beq
\psi^{\e,1}(\cdot,0)=\Gamma(\beta^{\e}(t_1)^*u_1,\e,\psi^{\e,1})(\cdot,0)
\label{ac} \eeq as well. Applying the projection $\Pi_c$ to both
\eqref{ab} evaluated at $t=0$ and \eqref{ac}, we see that indeed
$\beta^{\e}(t_1)=\beta^{\e}(t_1)^*$ as claimed.

An easy application of the implicit function theorem then reveals
the existence of a smooth curve of zeros $\e=\e(\beta^{\e})$ to the
equation \[
\frac{\e\beta^{\e}+\chi_{11}(\beta^{\e})^3+e(\beta^{\e},\e)}{\beta^\e}=0
\]
such that $\e=-\chi_{11}(\beta^{\e})^2+\mathcal{O}((\beta^{\e})^4)$.
Hence, there exist smooth curves of equilibria $\beta^{\pm}(\e)$ to
\eqref{stateode} for all small, positive $\e$ with
\beq\beta^{\pm}(\e)=\pm
\frac{1}{\sqrt{-\chi_{11}}}\,\e^{1/2}+\mathcal{O}(\e^{3/2}).\label{betaplus}\eeq
Consequently, within the collection of points on the center manifold
of the form $\Phi(\beta,\e)$ with $\beta$ real, the functions
$\psi_e^{\pm}:=\Phi(\beta^{\pm}(\e),\e)$ represent a supercritical
pitchfork bifurcation of equilibria from the normal state. The bound
\eqref{stable1} follows immediately from \eqref{verytang}. In light
of the rotation invariance of the problem, it immediately follows
that there is in fact a circle of equilibria
$e^{i\theta_0}\psi_e,\;\theta_0\in [0,2\pi)$ where we have written
simply $\psi_e$ for $\psi_e^+.$

Regarding stability of these equilibrium, it is clear from
\eqref{stateode} and the estimate \eqref{erest} that given any
initial data on the center manifold of the form $\Phi(\beta_0
u_1,\e)$ with $\beta_0$ real, positive and say bounded by
$C\sqrt{\e}$, the solution to \eqref{stateode} will converge to
$\beta^{+}(\e)$ and so the solution to \eqref{reform} will converge
to $\psi_e.$ Then since in light of \eqref{erot}, \eqref{stateode} is
clearly rotationally invariant, it follows that for complex initial
data on the center manifold, i.e. initial data of the form
$\Phi(\beta_0 u_1,\e)$ where $\beta_0=\abs{\beta_0}e^{i\theta_0}$
for some non-zero phase $\theta_0$, necessarily the solution will
converge to $e^{i\theta_0}\Phi(\beta^{+}(\e),\e)$. Thus, one
concludes that the circle of equilibrium states
$\{e^{i\theta_0}\psi_e:\;\theta_0\in [0,2\pi)\}$ is asymptotically
stable on the center manifold.

Finally, suppose that we start with initial conditions for
\eqref{reform} that are close to the circle of equilibria but that
do not lie on the center manifold. That is, suppose we have \beq
\label{closestart} \norm{\psi^{\e}(0)-e^{i\theta_{0}}\psi_e}_{H^1}
<r_1\quad\mbox{for some}\;\theta_{0}\in [0,2\pi)\eeq but that
$\psi^{\e}(0)\not\in\mathcal{M}_{\e}$. Without loss of generality,
we ignore this rotation for the remainder of the argument and take
$\theta_{0}=0$. We will argue that for $r_1$ and $\e$ sufficiently
small, again the trajectory $\psi^{\e}(t)$ is exponentially
attracted to $\psi_e$. We take in particular, $r_1<\frac{1}{4C_1} r$
where $C_1\geq 1$ and $r$ are the constants appearing in Theorem
\ref{exp}. Since $\norm{\psi_e}_{H^{1}}\leq
\frac{2}{\sqrt{-\chi_{11}}}\sqrt{\e}$ we can assert that
\beq\norm{\psi^{\e}(0)}_{H^{1}}<r/2\label{tri4}\eeq by choosing $\e$
sufficiently small and appealing to \eqref{closestart}.
 Denoting by
$\hat{\psi^{\e}}$ the trajectory on $\mathcal{M}_{\e}$, it then
follows from \eqref{expbd} and \eqref{closestart} that for as long
as $\psi(t)$ obeys the bound $\norm{\psi^{\e}(t)}_{H^1}<r$, one has
the estimate \beq \norm{\psi^{\e}(t)-\hat{\psi^{\e}}(t)}_{H^1}\leq
C_1 r_1 e^{-\eta t} <\frac r4 e^{-\eta t}.\label{psmall}\eeq

Then the triangle inequality to implies that \beq
\norm{\hat{\psi^{\e}}(0)-\psi_e}_{H^1}\leq (1+C_{1})r_{1}.
\label{tritwo} \eeq

Now $\hat{\psi^{\e}}(t)\in \mathcal{M}_{\e}$ is necessarily given by
$\hat{\psi^{\e}}(t)=\Phi(\beta^{\e}(t),\e)$ where $\beta^{\e}$ is
governed by \eqref{stateode}. As we already noted, up to a rotation
which we again ignore, the equilibrium value $\beta^{+}(\e)$ that is
stable under this flow, both in the sense of (exponential)
asymptotic approach, $\beta^{\e}(t)\to\beta^{+}(\e)$ and in the
sense that $\beta^{\e}(t)$ will stay close to this equilibrium for
all time. Choosing $r_{1}$ still smaller if necessary, we may appeal
to \eqref{tritwo} to conclude that
$\abs{\beta^{\e}(0)-\beta^{+}(\e)}$ is small and then the Lipschitz
property of the map $\Phi$ allows us to assert that, for instance,
\beq \norm{\hat{\psi^{\e}}(t)-\psi_e}_{H^1}< r/4\quad\mbox{for
all}\;t\geq 0.\label{tri3} \eeq

It follows from \eqref{tri4}, \eqref{psmall} and \eqref{tri3} that
in fact $\norm{\psi^{\e}(t)}_{H^{1}}<r$ for all $t\geq 0$ and so
\eqref{psmall} is valid for all time. Combining \eqref{psmall} with
the exponential approach of $\hat{\psi^{\e}}$ to $\psi_e$ along the
center manifold, we obtain the asymptotic stability of all
trajectories $\psi^{\e}(t)$ satisfying \eqref{closestart}.

The case where $I_k<I<I_c$ is handled similarly. Recall that in this
case, the parameter $\chi_{11}$ takes a positive value.
 Working then with $\e$ small and negative, and again
starting with initial data of the form $\Phi(\beta_0,\e)$ with
$\beta_0$ real, we find that \eqref{stateode} is now replaced by
\beq \beta^{\e}_t=
\e\beta^{\e}+\chi_{11}(\beta^{\e})^3+e(\beta^{\e},\e).
\label{unstateode}\eeq Another application of the implicit function
theorem reveals that \eqref{unstateode} possesses a pair of unstable
equilibria $\tilde{\beta}^{\pm}(\e)$ for all small negative
$\e$-values with $\tilde{\beta}^{\pm}(\e)=\pm
\frac{1}{\sqrt{\chi_{11}}}\,|\e|^{1/2}+\mathcal{O}(|\e|^{3/2})$. An
examination of \eqref{unstateode} shows that the corresponding
equilibria $\tilde{\psi}_e^{\pm}:=\Phi(\tilde{\beta}^{\pm}(\e),\e)$
bifurcating subcritically from the normal state are unstable as well
since even nearby PT-symmetric points on the center manifold, that
is points of the form $\beta u_1$ for $\beta$ real and near
$\beta^{\pm}(\e)$, flow away from them. Writing simply $\tilde{\psi}_e$
for $\tilde{\psi}_e^+$, the same is of course true for any of the
equilibria on the circle $\{e^{i\theta_0}\tilde{\psi}_e:\;\theta_0\in
[0,2\pi)\}.$

\end{proof}

We turn now to the case where the applied current $I$ satisfies
$I>I_c$ and a bifurcation to a periodic state occurs. We recall that
in this parameter regime the constants $\chi_{11}$ and $\chi_{12}$
are not real.
\begin{proposition}\label{Hopf}
Fix the applied current $I$  in the interval $(I_c,\infty)$. Then
provided that the constants $\hat{\chi}$ and $\tilde{\chi}$ given by
and \eqref{chihatdefn} and \eqref{chitildedefn}  respectively both
have positive real parts, the equation \eqref{reform} exhibits a
stable Hopf bifurcation to a periodic state $\psi_p=\psi_p(x,t,\e)$
branching from the normal state for all sufficiently small and
positive values of $\e$. This bifurcating solution, $\psi_{p}$ obeys
the estimate \beq
\norm{\psi_p-e^{i\theta_0}\left(\beta_{1,p}^\e(t)v_1+\beta_{2,p}^\e(t)v_2\right)}_{H^1}<C\e^{3/2}
\label{stable2}
 \eeq
for some $\theta_0\in [0,2\pi)$, where the real functions
$\beta_{1,p}^\e$ and $\beta_{2,p}^\e$ take the form \bea
&&(\beta_{1,p}^\e(t),\beta_{2,p}^\e(t))=
\left(\frac{\sqrt{\e}}{\sqrt{{\rm{Re}}\,\tilde{\chi}}}+\mathcal{O}(\e^{3/2})\right)
\nonumber\\
&&\times \bigg(\cos\left[\left(\imu+\frac{{\rm{Im}}\,\tilde{\chi}}{{\rm{Re}}\,\tilde{\chi}}\e+\mathcal{O}(\e^{3/2})\right)t\right],
\sin\left[\left(\imu+\frac{{\rm{Im}}\,\tilde{\chi}}{{\rm{Re}}\,\tilde{\chi}}\e+\mathcal{O}(\e^{3/2})\right)t\right]\bigg).
\nonumber\\
&&\label{ampe}\eea
\end{proposition} \begin{remark}\label{albe}{\rm Recalling the
relationship between the functions $v_1$ and $v_2$ given by
\eqref{vdefn}, one checks that the coefficients $\beta_{1,p}^\e$ and
$\beta_{2,p}^\e$ introduced above are related to the coefficients
$\alpha_1$ and $\alpha_2$ introduced in \eqref{h79} via the formulas
\[
\alpha_1(\e t)=\frac{1}{\sqrt{\e}}\,e^{i\imu t}\left(\beta_{1,p}^\e(t)+i
\beta_{2,p}^\e(t)\right),\quad \alpha_2(\e
t)=\frac{1}{\sqrt{\e}}\,e^{-i\imu t}\left(\beta_{1,p}^\e(t)-i
\beta_{2,p}^\e(t)\right)
\]
}\end{remark}
 \begin{remark} {\rm We recall from the previous section that numerically, we indeed find that
${\rm{Re}}\,\tilde{\chi}>0$ and ${\rm{Re}}\,\hat{\chi}>0$ for
$I>I_c$. See Figure \ref{chi_1}.}
\end{remark}

\begin{proof}
Invoking Theorem \ref{cm}, given any two complex numbers $\beta_1^0$
and $\beta_2^0$ of sufficiently small modulus, let $\psi^\e$ denote
the solution to \eqref{reform} subject to Dirichlet boundary
conditions and initial conditions given by
$\Phi(\beta_1^0,\beta_2^0,\e)$. Then we may describe $\psi^\e$ via
$\Phi$ at all future times as
$\psi^\e=\Phi(\beta_1^\e(t),\beta_2^\e(t),\e)$ for complex-valued
functions $\beta_1^\e(t)$ and $\beta_2^\e(t).$

We then project \eqref{reform} onto the center subspace and use
\eqref{veval} to obtain \bea &&
(\beta_1^\e)'v_1+(\beta_2^\e)'v_2=\nonumber\\ &&-\imu \beta_2^\e
v_1+\imu
\beta_1^\e v_2 +\e\beta_1^\e v_1+\e\beta_2^\e v_2\nonumber\\
&&+\bigg(\int_{-1}^{1}v_1\mathcal{N}\left[\beta_1^\e v_1+\beta_2^\e
v_2\right]\,dx\bigg)v_1+\bigg(\int_{-1}^{1}v_2\mathcal{N}\left[\beta_1^\e
v_1+\beta_2^\e v_2\right]\,dx\bigg)v_2\nonumber\\
&&+\bigg(\int_{-1}^{1}v_1\,\big(\mathcal{N}[\Phi(\beta_1^\e,\beta_2^\e,\e)]
-\mathcal{N}[\beta_1^\e v_1+\beta_2^\e
v_2]\big)\,dx\bigg)v_1\nonumber\\ &&+
\bigg(\int_{-1}^{1}v_2\,\big(\mathcal{N}[\Phi(\beta_1^\e,\beta_2^\e,\e)]
-\mathcal{N}[\beta_1^\e v_1+\beta_2^\e
v_2]\big)\,dx\bigg)v_2\nonumber.
 \eea
Integrating this equation first against $v_1$ and then against
$v_2$, we use the resulting two by two linear system in
$\beta_1^\e\,'$ and $\beta_2^\e\,'$ to find
\bea &&  \left(\begin{matrix}\beta_1^\e\nonumber\\
\beta_2^\e
\end{matrix}\right)'
= \left(\begin{matrix} \e&-{\rm{Im}}\,\lambda_1\nonumber\\
\imu&\e
\end{matrix}\right)
\left(\begin{matrix} \beta_1^\e\\ \beta_2^\e
\end{matrix}\right)+ \left(\begin{matrix}
 \int_{-1}^{1}v_1\mathcal{N}\left(\beta_1^\e v_1+\beta_2^\e
v_2\right)\,dx   \\
\int_{-1}^{1}v_2\mathcal{N}\left(\beta_1^\e v_1+\beta_2^\e
v_2\right)\,dx
\end{matrix}\right)\nonumber\\
&& + \left(\begin{matrix}
 \int_{-1}^{1}v_1\,\big(\mathcal{N}[\Phi(\beta_1^\e,\beta_2^\e,\e)]
-\mathcal{N}[\beta_1^\e v_1+\beta_2^\e
v_2]\big)\,dx      \\
\int_{-1}^{1}v_2\,\big(\mathcal{N}[\Phi(\beta_1^\e,\beta_2^\e,\e)]
-\mathcal{N}[\beta_1^\e v_1+\beta_2^\e v_2]\big)\,dx
\end{matrix}\right).
\nonumber
 \eea

Appealing to the center manifold estimate \eqref{verytang} and
carrying out a lengthy calculation similar to that of Section \ref{NP}, we
finally arrive at a system of the form
\bea &&  \left(\begin{matrix}\beta_1^\e\nonumber\\
\beta_2^\e
\end{matrix}\right)'
= \left(\begin{matrix} \e&-{\rm{Im}}\,\lambda_1\nonumber\\
\imu&\e
\end{matrix}\right)
\left(\begin{matrix} \beta_1^\e\\ \beta_2^\e
\end{matrix}\right)
+ \left(\begin{matrix} -[{\rm{Re}}\,\tilde{\chi}
(R^\e)^2+i{\rm{Im}}\,\hat{\chi}\gamma^\e]\beta_1^\e+
[{\rm{Im}}\,\tilde{\chi} (R^\e)^2-i{\rm{Re}}\,\hat{\chi}\gamma^\e]\beta_2^\e\\
-[{\rm{Im}}\,\tilde{\chi}
(R^\e)^2-i{\rm{Re}}\,\hat{\chi}\gamma^\e]\beta_1^\e
-[{\rm{Re}}\,\tilde{\chi} (R^\e)^2+i{\rm{Im}}\,\hat{\chi}\gamma^\e]
\beta_2^\e\end{matrix}\right)\nonumber\\
&&\label{bsys}\qquad\qquad\qquad\qquad\qquad+\mathcal{O}\left(\e(R^\e)^3+(R^\e)^5\right).
 \eea In the system above we have introduced the notation
 \beq
R^\e:=\sqrt{\abs{\beta_1^\e}^2+\abs{\beta_1^\e}^2}\quad\mbox{and}\quad
\gamma^\e:=i\left((\beta_1^\e)^*\beta_2^\e-\beta_1^\e(\beta_2^\e)^*\right)=
\abs{\beta_1^\e}\abs{\beta_2^\e}\sin(\theta_1^\e-\theta_2^\e),\label{Rgam}
\eeq where $\beta_j^\e=\abs{\beta_j^\e}e^{i\theta_j^\e}.$

We now apply the standard method for proving the existence of a
periodic solution to \eqref{bsys} via a Hopf bifurcation and to
compute rigorously the amplitude and period of the oscillations. To
this end, we consider first the case where the initial values
$\beta_1^0$ and $\beta_2^0$ are real. Then the resulting
PT-symmetric initial data $\Phi(\beta_1^0,\beta_2^0,\e)$ for
\eqref{reform} will lead to the PT-symmetry of the solution at all
future times. Consequently, for all $t>0$, the projection of
$\psi^\e$ onto the center subspace must take the form
$\beta_1^\e(t)v_1+\beta_2^\e(t)v_2$ where $\beta_1^\e$ and
$\beta_2^\e$ are real. This leads to a significant simplification of
\eqref{bsys} in that $\gamma^\e\equiv 0.$

Converting to polar coordinates, $R^\e$ and
$\theta^\e:=\tan^{-1}\left(\beta_2^\e/\beta_1^\e\right)$, we derive
the system \bea && (R^\e)'=\e R^\e-{\rm{Re}}\,\tilde{\chi}
(R^\e)^3+\mathcal{O}\left(\e(R^\e)^4+(R^\e)^6\right),\label{Reqn}\\
&& (\theta^\e)'=\imu -{\rm{Im}}\,\tilde{\chi}
(R^\e)^2+\mathcal{O}\left(\e(R^\e)^2+(R^\e)^4\right).\label{thetaeqn}
\eea Starting with any initial condition with small amplitude
$a:=R^\e(0)$, it is easy to argue that $\theta^\e$ is a monotone
increasing function of $t$, thus justifying the description of
$R^\e$ as $R^\e(a,\theta^\e)$ via the scalar O.D.E. \beq \frac{d
R^{\e}}{d \theta^\e}=
\frac{\e}{\imu}R^\e-\frac{{\rm{Re}}\,\tilde{\chi}}{\imu}
(R^\e)^3+g(\e,R^\e,\theta^\e)\label{scalarODE}\eeq where
\[
g(\e,R^\e,\theta^\e)=\mathcal{O}\left(\e (R^\e)^3+(R^\e)^5\right).
\]

A periodic orbit for \eqref{bsys} corresponds to a value of the
amplitude $a$ such that $R^\e(a,2\pi)=a.$ Thus, we seek a fixed
point of the Poincar\'e return map $\Pi$ given by \beq \Pi(\e,a):=
ae^{2\pi\e/\imu}+e^{2\pi\e/\imu}\int_0^{2\pi} e^{-(\e/\imu)\theta}
\left(\frac{-{\rm{Re}}\,\tilde{\chi}}{\imu}
R^\e(a,\theta)^3+g(\e,R^\e(a,\theta),\theta)\right)\,d\theta.
\label{poincare}\eeq

Using \eqref{scalarODE} one readily checks that \beq
R^\e(a,\theta)=a+\mathcal{O}(\e a+a^3)\quad\mbox{and so}\quad
R^\e(a,\theta)^3=a^3+\mathcal{O}(\e a^3+a^5).\label{aeest}\eeq
Therefore, expanding $\Pi$ for small $\e$ we see that
\[
\Pi(\e,a)=a+\frac{2\pi\e}{\imu}a-\frac{2\pi{\rm{Re}}\,\tilde{\chi}}{\imu}a^3
+ \mathcal{O}(a\e^2+\e a^3+a^5)\] Via the implicit function theorem
one derives a curve of zeros $\e(a)$ for the expression
\[ \frac{\Pi(\e,a)-a}{a}.\] Hence, we obtain a curve of fixed points of $\Pi$
 with \[ \e(a)=({\rm{Re}}\,\tilde{\chi})
a^2+\mathcal{O}(a^3),\] or inverting this relationship, \beq
a=\frac{1}{\sqrt{{\rm{Re}}\,\tilde{\chi}}}\,\e^{1/2}+\mathcal{O}(\e).\label{iftcurve}\eeq
Denoting the resulting periodic solution to \eqref{Reqn}-\eqref{thetaeqn} by $(R_{p}^\e(t),\theta_{p}^\e(t))$
and letting $(\beta_{1,p}^\e,\beta_{2,p}^\e)=(R_{p}^\e\cos(\theta_p^\e), R_{p}^\e\sin(\theta_p^\e))$ denote
the corresponding periodic solution to \eqref{bsys},
the asymptotic estimates \eqref{stable2} and \eqref{ampe} then
readily follow from \eqref{thetaeqn}, \eqref{aeest} and
\eqref{iftcurve}.

We also note that the asymptotic stability of this periodic orbit,
and hence the asymptotic stability of the periodic solutions to
\eqref{reform} among nearby PT-symmetric competitors on the center
manifold, is a consequence of the condition
\[
\frac{\partial\Pi}{\partial a}(\e(a),a)\sim
1-\frac{4\pi{\rm{Re}}\,\tilde{\chi}}{\imu}a^2<1,
\]
cf. \cite{HK}, Thm. 12.13. In light of the rotational invariance of
the system \eqref{bsys}, this means not only that real initial data
$(\beta_1^\e(0),\beta_2^\e(0))$ lying sufficiently close to the
periodic orbit will be asymptotically drawn by the flow into the
orbit but also that any rotation, say
$e^{i\theta_0}(\beta_1^\e(0),\beta_2^\e(0))$ of such initial data
will be drawn to the corresponding rotation of this orbit as well.

We next wish to argue that these periodic orbits are stable within
the class of all flows starting nearby on the center manifold. For
this purpose we return to \eqref{bsys}; that is, we consider the
situation where the initial data for  $\beta_1^\e$ and $\beta_2^\e$
are not necessarily real. Note that this corresponds to initial data
that is then not assumed to be PT-symmetric. We claim that if we
start with complex initial data $(\beta_1^\e(0),\beta_2^\e(0))$
sufficiently close to any of the orbits given by
\eqref{stable2}-\eqref{ampe}, then again the flow will draw the
resulting solution into one of these periodic orbits.

To establish this claim it is useful to derive a system of O.D.E.'s
for the quantities $(R^{\e})^2$ and $\gamma^\e$. Differentiating the
defining formulas in \eqref{Rgam} and using the system satisfied by
$\beta_1^\e$ and $\beta_2^\e$ given in \eqref{bsys}, a routine
calculation leads us to:
 \beq
\left(\begin{matrix}(R^{\e})^2\\ \gamma^\e\end{matrix}\right)'=
\left(\begin{matrix} 2\left[\big(\e-{\rm{Re}}\,\tilde{\chi}\,
(R^\e)^2\big)(R^\e)^2-{\rm{Re}}\,\hat{\chi}(\gamma^\e)^2\right]\\
2\left[\e-{\rm{Re}}\,(\tilde{\chi}+\hat{\chi})
(R^\e)^2\right]\gamma^\e
\end{matrix}\right)\quad+\quad\mathcal{O}(\e(R^\e)^4+(R^\e)^6).\label{Rg}
\eeq Note that $\beta_j^\e$ are parallel if and only if
$\gamma^\e=0$. Recalling that PT-symmetric solutions correspond to
$\beta_j^\e$ real, we thus see that $\gamma^\e=0$ corresponds to the
invariant three-dimensional manifold of PT-symmetric solutions and
their complex rotations that we have already considered, with the
single remaining dimension parametrized conveniently by $\gamma^\e$.

If we momentarily ignore the error term, then we can readily
identify an equilibrium point for \eqref{Rg} at
$(\e/{\rm{Re}}\,\tilde{\chi},0)$. Linearization about this critical
point immediately yields linear stability and in fact a standard
phase plane analysis yields nonlinear asymptotic stability. To treat
the full system \eqref{Rg} (that is, including the error term), we
must linearize instead about the nearby periodic solution
$((R_p^\epsilon)^2,0)$, corresponding to the (real) periodic orbit
$(\beta^\e_{1,p},\beta^\e_{2,p})$ constructed above, yielding
$$
(\gamma^\epsilon )'=
2\left[\e-{\rm{Re}}\,(\tilde{\chi}+\hat{\chi}) (R^\e_p)^2
+\mathcal{O}(\e(R^\e_p)^2+(R^\e_p)^4) \right] \gamma^\e.
$$
Since the error term is clearly lower order, the condition
$(R_p^\e)^2\approx \e/{\rm{Re}}\,\tilde{\chi}$  again implies
exponential linearized stability in the $\gamma^\epsilon$ direction.
As exponential orbital stability has already been verified in the remaining directions,
we may conclude exponential linearized and nonlinear orbital stability of the
family of periodic orbits within the full four-dimensional center manifold
and not only within the invariant three-dimensional $\gamma=0$ manifold.

It remains to argue that the periodic orbits attract nearby initial
data for \eqref{reform} off of the center manifold. The argument for
this fact follows exactly as did the corresponding point in the
proof of Proposition \ref{stat} regarding stability of stationary
solutions through an appeal to Theorem \ref{exp}.
\end{proof}

Proposition \ref{Hopf} describes the bifurcating stable periodic
solutions of main physical interest. With a little further effort,
we may also rigorously obtain the existence of the two unstable
periodic branches derived formally at the end of Section \ref{NP} and
obtain essentially a complete description of the full
four-dimensional bifurcation.

\begin{proposition}\label{Hopf2}
Under the assumptions of Proposition \ref{Hopf}, the equation
\eqref{reform} also exhibits an unstable
bifurcation for small $\eps>0$
to periodic states $\psi_{p,\pm}=\psi_{p,\pm}(x,t,\e)$ of form
%for all small positive $\e$:
\beq
\norm{\psi_{p,\pm}-e^{i\theta_0}\left(\beta_{1,\pm}^\e(t)v_1+\beta_{2,\pm}^\e(t)v_2\right)}_{H^1}<C\e^{3/2},
\label{unstable2}
 \eeq
where, recalling that
$\chi_{11}=-\frac{1}{2}(\tilde{\chi}+\hat{\chi})$ so that
$-{\rm{Re}}\, \chi_{11}>0$,
\begin{equation}
\beta_1^{\e,\pm}(t)=
\left(\frac{\sqrt{\e}}{2\sqrt{-{\rm{Re}}\,{\chi}_{11}}}+\mathcal{O}(\e^{3/2})\right)
e^{\pm i\big((\imu
+\e\frac{{\rm{Im}}\,\chi_{11}}{{\rm{Re}}\,\chi_{11}})t+\mathcal{O}(\e^{3/2})\big)}
\quad\mbox{and}\quad\beta_2^{\e,\pm}=\pm
i\beta_1^{\e,\pm}.\label{unstable3}\end{equation}
%For a positive value $a$ independent of $\e$,
%these two states, along with the stable periodic
%state $\psi_p$ constructed in Theorem \ref{Hopf} and the normal
%state, represent the only
%persistent states on the four-dimensional
%manifold $\mathcal{M}_\e$, cf. \eqref{cmg}, within $B(0,a)$.

Furthermore, there exists a positive value $a$ independent of $\e$
such that for all small $\e>0$, these two states, along with the
stable periodic state $\psi_p$ constructed in Theorem \ref{Hopf} and
the normal state, represent the only persistent states in
$B(0,a):=\{\psi:\;\norm{\psi}_{H^1}<a\}$. Indeed, the phase portrait
within $B(0,a)$ consists of two repelling codimension two $C^\infty$
stable manifolds $\tilde \cN_{p,\pm}$ of the unstable periodic
solutions $\psi_{p,\pm}$, a repelling codimension four $C^\infty$
stable manifold $\tilde \cN_0$ of the normal state $\psi=0$, and an
attracting codimension one $C^\infty$ stable manifold $\tilde
\cN_{PT}$ of the invariant manifold $\cN_{PT}$ of PT-symmetric
solutions and their rotations
lying within $\cMe\cap B(0,\overline{C}a)$, cf. Remark
\ref{foliation}. The latter contains both the unstable normal
equilibrium and the stable periodic solutions, with all other
solutions flowing from repelling submanifolds $\tilde \cN_{p,\pm}$
and $\tilde \cN_0$ to the attracting submanifold $\tilde \cN_{PT}$
and thereafter to the stable periodic solutions.

In particular, all solutions originating in $B(0,a)$ converge
either to the stable periodic, an unstable periodic, or the normal state.
For generic initial data $\psi_0\in B(0,a)$, the corresponding solution
converges to the stable periodic state, the only exceptional data
lying on the codimension two and four submanifolds $\tilde \cN_{p,\pm}$
and $\tilde \cN_0$, respectively.
%
% Finally, for $\epsilon>0$, the phase portrait within
%$B(0,a)\cap\mathcal{M}_\e$ contains three invariant submanifolds:
%two repelling two-dimensional submanifolds $\cN_{p,\pm}$
%consisting of the stable manifolds of the unstable periodic solutions,
%and the attracting three-dimensional manifold $\cN_{PT}$
%of PT-symmetric solutions and their complex rotations,
%the latter containing both the unstable normal
%equilibrium and the stable periodic solution, with all other
%solutions flowing from repelling two-dimensional submanifolds to the
%attracting three-dimensional submanifold.
\end{proposition}

\begin{remark}
{\rm The description \eqref{unstable2}-\eqref{unstable3} of the
unstable periodic solutions given in $\beta$-coordinates may be
recognized as profile \eqref{a2unst}-\eqref{a1unst} given in the
$\alpha$-coordinates of Section \ref{NP}.
%TODO: section ref., not number.
}
\end{remark}
\begin{proof}
Our approach here is generally to establish the claims of the
theorem first for the system of differential equations \eqref{bsys}
or \eqref{Rg} ignoring higher order error terms and then to broaden
the claims to the full equations including the perturbations. Hence,
we first note that in light of the smoothness of these higher order
terms and the polynomial error bounds, the perturbations are seen to
be small in the $C^1$ topology in a small neighborhood of
%origin,
the origin,
i.e. for $\beta^\e_1,\,\beta^\e_2$ small or equivalently, for $R^\e$
small.
%CHANGED (important addition)-K:
%Moreover, on the intersection of
%$\mathcal{M}_\e$ and $B(0,C\sqrt{\e})$, where the unperturbed flow
%is small of order $\eps$, it is easily seen that perturbations
%are also small in $C^1$ topology for the rescaled,
%slow variables $\beta_j\to \beta_j/\sqrt{\eps}$, $t\to t/\eps$,
%for which flow is order one.
%ENDCHANGED

We begin our analysis by turning to the system \eqref{Rg}. For
convenience, let us denote the quantity $(R^\e)^2$ by $A^\e$ and so
view this system as one for $\gamma^\e$ given by \eqref{Rgam} and
$A^\e:=\abs{\beta^\e_1}^2+\abs{\beta^\e_2}^2.$ We list some simple
observations:

\par\noindent $\bullet$ Elementary phase plane analysis of this system reveals
that there exists a small positive number $r$, independent of $\e$,
such that the triangle
\[T:=\{(\gamma^\e,A^\e):\;0\leq \abs{\gamma^\e}\leq A^\e\leq r\}\]
is invariant for the flow \eqref{Rg} for all $\e$ sufficiently
small. That the flow cannot exit the two bottom sides of this
triangle is trivial in light of the algebra
\[\abs{\gamma^\e}\leq 2\abs{\beta^\e_1}\abs{\beta^\e_2}\leq
\abs{\beta^\e_1}^2+\abs{\beta^\e_2}^2=A^\e.\] That the flow cannot
exit the top where $A^\e=r$ follows immediately from the fact that
\begin{equation}(A^\e)'\leq 2\bigg(\e-{\rm{Re}}\,\tilde{\chi}(A^\e)+C\e
(A^\e)+C(A^\e)^2\bigg)A^\e<0\label{down}\end{equation} provided that
$\e<<A^\e$ and $A^\e=r$ is small. Here $C$ is a positive constant
coming from the error bounds in \eqref{Rg}.

\par\noindent
$\bullet$ A further consequence of \eqref{down} for the full system
\eqref{Rg} and its ``parent" system \eqref{bsys} is that if one
starts with data lying in the triangle $T$, then the corresponding
orbit will exponentially approach a sub-triangle in which
$A^\e=\mathcal{O}(\e).$
%CHANGED (added this):
The rate of approach may be obtained by direct computation from the
differential inequality $(A^\e)'\le -C_1(A^\e)^2$, $C_1>0$.
%ENDCHANGED

\par\noindent
$\bullet$ If one ignores the error terms in this system, then it is
easy to check that there are exactly four critical points located
at\[
(\gamma^\e,A^\e)=(0,\frac{\e}{{\rm{Re}}\,\tilde{\chi}}),\;(0,0)\quad\mbox{
and}\quad \frac{\e}{2\abs{{\rm{Re}}\,\chi_{11}}}(\pm 1,1)\] and the
Jacobian at each of these points is nonzero (see the next item
below).
%REMOVED (Peter, there IS not full system in the triangle- at least
%I don't think so. For, the error terms depend on variables other than
%A and \gamma it seems....
%That's why I had to change this part significantly. -K
%
%Since the error terms in \eqref{Rg} are $C^1$-small, it
%follows that there exists a critical point of the full system near
%each of these four points and that there are no other critical
%points in $T$.
%%CHANGED (the above is not adequate)-K:
%Here, we are using the second observation that errors are $C^1$-small
%also in the blown-up ``slow'' variables, or, equivalently, that errors
%are small even compared to the $O(\eps)$ size of the eigenvalues of
%the full system at these critical points.
%%ENDCHANGED

\par\noindent
$\bullet$ Linearizing \eqref{Rg} about these four critical points,
one finds: (i) The critical point near
$(0,\frac{\e}{{\rm{Re}}\,\tilde{\chi}})$ is stable, with two
negative eigenvalues, as has been noted already in the proof of
Theorem \ref{Hopf}. This point corresponds to the stable periodic
orbit arising from a Hopf bifurcation. (ii) The normal state $(0,0)$
is unstable with a multiplicity two positive eigenvalue. (iii)
$\frac{\e}{2\abs{{\rm{Re}}\,\chi_{11}}}(\pm 1,1)$ are both saddles
with one positive and one negative eigenvalue. The negative
eigenvalue corresponds to the eigenvector $(\pm 1,1).$

\par\noindent
$\bullet$ Not only do $(\pm 1,1)$ correspond to the stable
eigendirections of the linearization about
$\frac{\e}{2\abs{{\rm{Re}}\,\chi_{11}}}(\pm 1,1)$ respectively, but
in fact, the bottom sides of the triangle $T$, i.e. $\gamma^\e=\pm
A^\e$ are exactly the stable manifolds of these unstable critical
points for the system \eqref{Rg} ignoring error terms.

\par\noindent
$\bullet$ Introducing $s$ as the ratio $\gamma^\e/A^\e$, we use
\eqref{Rg} to calculate that
\begin{equation} s'=-2{\rm{Re}}\,\hat{\chi} A^\e
s(1-s^2)+\mathcal{O}(\e A^\e+(A^\e)^2).\label{seqn}\end{equation} As
a consequence, if one ignores error terms, then one sees the flow
repels away from the triangle's lower sides $s=\pm 1$ towards the
manifold of $PT$-symmetric solutions and their rotations, $s=0$,
i.e. towards $\gamma^\e=0.$
We should note that since the flow
\eqref{reform} preserves PT-symmetry, it follows that $s=0$ must in
fact be a critical point of \eqref{seqn} even with inclusion of
error terms.

Having noted these simple properties of \eqref{Rg}, we next
establish the existence of the unstable periodic orbits for
\eqref{bsys} and hence for \eqref{reform}. As we shall see, these
correspond precisely to the two unstable critical points near
$\frac{\e}{2\abs{{\rm{Re}}\,\chi_{11}}}(\pm 1,1)$ found above in the
$\gamma^\e A^\e$ plane. Let us first note that back in Section
\ref{NP}, we identified exact periodic solutions on the invariant
manifolds $\alpha_1\equiv 0$ and $\alpha_2\equiv 0$ for the
equations \eqref{h85a}-\eqref{h85b} obtained by neglecting
higher-order error terms. In light of the relations \eqref{albe}
this corresponds to exact solutions (again neglecting error terms)
to \eqref{bsys} found along the manifold $\beta^\e_2=\pm
i\beta^\e_1$. This, in turn, corresponds to solutions of \eqref{Rg}
for which $\gamma^\e=\pm A^\e$.

The explicit formulas for these two unstable periodic solutions can
be obtained by directly solving \eqref{bsys} under the constraint
$\beta^\e_2=\pm i\beta^\e_1$ and are given in \eqref{unstable3}.
Note in particular that the two equations in \eqref{bsys} for
$(\beta^\e_1)'$ and $(\beta^\e_2)'$ reduce to the same equation
under this constraint. Their saddle-type instability follows from
the observations listed above and in particular from \eqref{seqn}.

In order to argue that these saddle-type periodic solutions persist
with the inclusion of error terms in \eqref{bsys}, we define
$\omega^\e_1:=(\beta^\e_1+i\beta_2^\e)$ and
$\omega^\e_2:=(\beta^\e_1-i\beta_2^\e)$ so that $\omega^\e_1$ and
$\omega^\e_2$ are closely related to the $\alpha_j$ variables of
Section 5 via the formulas $\omega^\e_j=\sqrt{\e}e^{(-1)^j\imu
t}\alpha_j$ for $j=1,2.$ We will describe how to find a periodic
solution to the full system \eqref{bsys} nearby the solution
corresponding to $\beta^\e_1=i\beta^\e_2$, (i.e. $\omega^\e_2=0$) for the
truncated system. The same analysis can be repeated to find the
other unstable solution to the full system as well.

To this end, it is perhaps best to introduce a Poincar\'e map as
follows: Consider \eqref{bsys} without error terms and start with
initial data $\omega^\e_1(0),\;\omega^\e_2(0)$ such that
${\rm{Im}}\,\omega^\e_1(0)=0$ while ${\rm{Re}}\,\omega^\e_1(0)$ is
near $\sqrt{\e}/\sqrt{-{\rm{Re}}\,\chi_{11}}$ and
$\abs{\omega^\e_2(0)}$ is small. We then measure the values of
${\rm{Re}}\,\omega^\e_1$, ${\rm{Re}}\,\omega^\e_2$ and
${\rm{Im}}\,\omega^\e_2$ at the next time, say
$T=T(\omega^\e_1(0),\omega^\e_2(0),\e),$ at which the trajectory
$(\omega^\e_1(t),\omega^\e_2(t))$ crosses the $3$-plane
$\{(\omega_1,\omega_2)\in \C\times\C:\;{\rm{Im}}\,\omega_1=0\}.$ Let
us denote this map
\begin{equation}({\rm{Re}}\,\omega^\e_1(0),{\rm{Re}}\,\omega^\e_2(0),
{\rm{Im}}\,\omega^\e_2(0))\rightarrow
({\rm{Re}}\,\omega^\e_1(T),{\rm{Re}}\,\omega^\e_2(T),
{\rm{Im}}\,\omega^\e_2(T))\label{returnmap}\end{equation} by
$F^\e:\R^3\to\R^3.$ Note that
$(\frac{\sqrt{\e}}{\sqrt{-{\rm{Re}}\,\chi_{11}}},0,0)$ is a fixed
point of this map corresponding to the exact unstable periodic
solution of the truncated system satisfying
$\beta^\e_1=i\beta^\e_2$.

By converting \eqref{bsys} into a system for $\omega^\e_j$, $j=1,2$
one can carry out a tedious but straight-forward calculation to
determine that when the Jacobian matrix of $F^\e$ is evaluated at
the fixed point
$(\frac{\sqrt{\e}}{\sqrt{-{\rm{Re}}\,\chi_{11}}},0,0)$, it takes the
form
%\[
\begin{equation}\label{DF}
DF^\e=\left(\begin{matrix}
\mu^\e&0&0\\
0&e^{a^\e}\cos{b^{\e}}&-e^{a^\e}\sin{b^{\e}}\\
0&e^{a^\e}\sin{b^{\e}}&e^{a^\e}\cos{b^{\e}}
\end{matrix}\right).
\end{equation}
%\]
Here
\[
\mu^\e=e^{-2\e p^\e},\quad
a^{\e}=\e\left(\frac{{\rm{Re}}\,\hat{\chi}}{-{\rm{Re}}\,\chi_{11}}\right)p^\e,\quad
b^{\e}=\left(\imu-\e
\frac{{\rm{Im}}\,\chi_{12}}{{\rm{Re}}\,\chi_{11}}\right)p^\e,\\
\]
and
\[
p^\e=\frac{2\pi}{\imu+\frac{{\rm{Im}}\,\chi_{11}}{{\rm{Re}}\,\chi_{11}}\e},\]
which we recognize as the period of the unstable periodic solution
with error terms ignored.

Note that there is one eigenvalue of $DF^\e$ of modulus less than
one, namely $\mu^\e$, giving one stable direction to the map, while
the other two unstable eigenvalues $e^{a^{\e}\pm i b^\e}$ have
modulus greater than one since $a^\e>0$. Thus, one again sees the
saddle-type instability of this periodic solution from the
perspective of this Poincar\'e map.

To argue that this picture persists under perturbation (i.e. under
inclusion of the error terms in \eqref{bsys}), we observe that the
three eigenvalues of $D(F^\e-I)$ take the form
\begin{equation}
-\frac{4\pi}{\imu}\e+\mathcal{O}(\e^2),\quad
\frac{2\pi}{(\imu)\,(-{\rm{Re}}\,\chi_{11})}\bigg({\rm{Re}}\,\hat{\chi}\pm
i{\rm{Im}}\,\tilde{\chi}\bigg)\e +\mathcal{O}(\e^2).
\label{smalleig}\end{equation} Once we incorporate the error terms
from \eqref{bsys} in our analysis, we use the fact that within the
$\mathcal{O}(\sqrt{\e})$ ball where we are working, these
$\mathcal{O}(\e|\beta|^3+ |\beta|^5)$ error terms are necessarily
$\mathcal{O}(\e^{5/2})$ with a contribution to the entries of the
original Jacobian matrix $DF^\e$ of at most
$\mathcal{O}(\e|\beta|^2+ |\beta|^4)=\mathcal{O}(\e^2)$. Hence,
again subtracting the identity to form the displacement map of the
full system including error terms, \eqref{smalleig} guarantees that
the differential of this map is still nonsingular. Necessarily it then must still
map a neighborhood of the original critical point onto the origin, that is,
there must still exist a fixed point of the map corresponding to $F$ perturbed by
higher order terms. What is more, the Jacobian matrix of this perturbed map evaluated
at the fixed point must still possess one eigenvalue of modulus less than one and two
of modulus greater than one so the saddle-type instability persists under
perturbation for these periodic orbits.

It remains to establish the asserted global behavior of solutions
originating in the $H^1$-ball $B(0,a).$
We begin by describing the behavior in $\cMe\cap B(0,\overline{C}a)$,
where $\overline{C}>0$ is the geometric constant of Remark \ref{foliation},
to be used later.
As observed previously, the manifold $\cN_{PT}\subset \cMe$
of PT-symmetric solutions and their rotations remains invariant
for the full as well as the unperturbed equations, and is locally
attracting by \eqref{seqn}.
Likewise, the normal state ($\beta^\e_1=\beta^\e_2=0$) is an equilibrium of the full system
that is locally
repelling.
Since the unstable periodic solutions
correspond to hyperbolic, saddle-type equilibria of the associated return map,
with a single stable eigenvalue and a pair of unstable eigenvalues,
we find that they possess, for fixed $\eps$, a $C^\infty$ stable manifold of
dimension one of the return map. Under time-evolution, this induces
stable manifolds $\cN_{p,\pm}$
of dimension two within the center manifold, as claimed.

The structure and attracting or repelling properties of $\cN_{p,\pm}$
are difficult to determine outside an $\eps$-neighborhood of the periodic
solution, due to the fast, order one, angular flow relative to the
order $\eps$ flow measured by the return map. Thus, it would appear that a determination
of the global
structure in $\cMe \cap B(0,Ca)$ would require more complicated
averaging arguments outside the scope of the present analysis.
However, we may finesse this point using more elementary tools
together with the special structure of our
%TODO: spelling of finess?
equations.

We first observe by \eqref{seqn} that for $C>0$ sufficiently large,
the neighborhoods $\cK_\pm$ of the triangle $T$ given by
\begin{equation}\label{Kdef}
\cK_\pm:=\{(\gamma^\e,A^\e)\in
T:\;s:=\gamma^e/A^\e\;\mbox{satisfies}\;\, |s\pm 1|\le C(\eps +
A^\e)\}
\end{equation}
are invariant under backward
flow of the ODE on the center manifold. Hence, by estimate \eqref{unstable2}-\eqref{unstable3}
showing that $||s|-1|\le C\eps$ along the unstable periodics,
the stable manifolds $\cN_{p,\pm}$
of the unstable periodics are confined to
$\cK_\pm$. In particular, they lie no more than $C(\eps + A^\e)$ from the sides of the triangle
corresponding in $\gamma^\e$-$A^\e$ coordinates to the stable manifolds
of the unstable periodic solutions of the unperturbed equations.
Moreover, from \eqref{Rg}, we see that for $(\gamma^e,A^\e)\in\cK_\pm$
one has
$$
(A^\e)'= 2(\e-{\rm{Re}}\,(\tilde{\chi}+\hat{\chi})A^\e)A^\e
+\mathcal{O}(\eps + A^\e)(A^\e)^2.
$$
Thus, for $(\gamma^\e,A^\e)\in \cK_\pm$ outside an $\eps^2$-neighborhood of
the critical point $\frac{\e}{{\rm{Re}}\,(\tilde{\chi}+\hat{\chi})}(\pm 1,1)$
of the unperturbed $\gamma^\e$-$A^\e$ system, $A^\e$ is strictly decreasing
for $A^\e>\frac{\e}{{\rm{Re}}\,(\tilde{\chi}+\hat{\chi})}$ and strictly decreasing for
$A^\e<\frac{\e}{{\rm{Re}}\,(\tilde{\chi}+\hat{\chi})}.$ Hence in these
regions, $\cN_{p,\pm}$ are graphs over their unperturbed counterparts
and $\mathcal{O}(\e + A^\e)$ close to them.

%CHANGED
%*******************************NEW VERSION*******

It remains to treat the excluded $\eps^2$ neighborhood in
$\gamma^\e$-$A^\e$ coordinates of the critical point
$\frac{\e}{{\rm{Re}}\,(\tilde{\chi}+\hat{\chi})}(\pm 1,1)$
of the unperturbed $\gamma^\e$-$A^\e$ system.
By a straightforward computation, we find that this
neighborhood is contained in the image of an $\eps$-neighborhood
in $\beta^\e$-coordinates of the orbit of the order $\eps^{1/2}$ amplitude
unstable periodic solution \eqref{unstable3} of the
unperturbed $\beta^\e$-equations.

Indeed, a brief examination in $\beta^\e$-coordinates shows that the
$\gamma^\e$-$A^\e$ estimate is overly conservative (a result of the
singularity in certain directions of the coordinate change between
the two coordinate systems). Changing to the more convenient
$\omega$-coordinates, for which $|\omega^\e_1|=
\frac{\sqrt{\e}}{\sqrt{-{\rm{Re}}\,\chi_{11}}}$, $\omega_2^\e=0$
corresponds to the unstable periodic orbit, and $\omega_2^\e=0$ its
stable manifold, we find compute that
$$
\begin{aligned}
|\omega_1^\e|_t&=
(\eps +({\rm{Re}}\,\chi_{11} |\omega^\e_1|^2 + {\rm{Re}}\, \chi_{12} |\omega^\e_2|^2)
|\omega_1^\e| +O(\eps |\omega^\e|^3+|\omega^\e|^5) ,\\
|\omega_2^\e|_t&=
(\eps +({\rm{Re}}\,\chi_{11} |\omega^\e_2|^2 + {\rm{Re}}\, \chi_{12} |\omega^\e_1|^2)
|\omega_2^\e| +O(\eps |\omega^\e|^3+|\omega^\e|^5)\\
\end{aligned}
$$
(cf. \eqref{mh85a}--\eqref{mh85b}). Using this system, we can check
that as long as the deviations $\tilde{\omega}^\e_1$ and
$\tilde{\omega}^\e_2$ from the unstable periodic solution to the
unperturbed system satisfy
\[\e^{3/2}<<\abs{\tilde{\omega}^\e_1}<<\e^{1/2}\quad\mbox{and}\quad
\e^{3/2}<<\abs{\tilde{\omega}^\e_2}<<\e^{1/2},\] the stable
manifolds $\cN_{p,\pm}$ are graphs over their unperturbed
counterparts ($\omega^\e_2\equiv 0$ in this case) and all other
solutions are repelled toward $\cN_{PT}$ since
$\frac{d}{dt}\abs{\tilde{\omega}^\e_1}$ will be negative in this
regime while $\frac{d}{dt}\abs{\tilde{\omega}^\e_2}$ will be
positive.
%TODO: hmmm, do we need this after all?  There seems to be
%a large margin for error in the arguments..  maybe not..

%CHANGED (not really, but just for you to find this note...!)
%
%gamma-A coordinates DO turn out as I feared to be slightlhy problematic
%due to above-mentioned singularity of coordinate change, but only at this one
%point in the argument...  In the end, a small price to pay for
%the very simple picture in unperturbed case, maybe...
%though in retrospect, $|\omega_1|$ %and $|\omega_2|$ give an
%equally simple planar system with same properties, and nonsingular
%coord. change, or at least less %singular (bi-Lipshitz, so distances
%preserved, avoiding above difficulty).
%
%This little snag is what took me so long, by the way...
%in the end, just a confusion, since it turned out there
%was plenty of wiggle room in the argument after all...
%micro-description turns out to be on order eps^{1/2} ball
%as of course it should be from our formal asymptotics...
%--K
%ENDCHANGED

On the other hand, by our previous estimates on the orders of
perturbation terms and their Jacobians in the return map $F^\e$
given by \eqref{returnmap}, the Taylor expansion of this map at the
unstable fixed point  corresponding to the unstable periodic,
denoted here by $p^\e_*$, is
$$
\tilde{F}^\e(z):=F^\e(p^\e_*+z)- p^\e_* = dF^\e(p^\e_*)z + N^\e(z,z)
+ \Theta^\e(z).
$$
Here the quadratic order Taylor remainder term $N^\e$ is order
\begin{equation}\label{Nbd}
N^\e=\mathcal{O}(|p^\e_*||z|^2),\quad
N^\e_z=\mathcal{O}(|z|^2+|p^\e_*||z|)
\end{equation}
and the term $\Theta^\e$ which incorporates the perturbation term is
order
\begin{equation}\label{Thetabd}
\Theta^\e=\mathcal{O}(\e|p^\e_*|^3 +|p^\e_*|^5),\quad
\Theta^\e_z=\mathcal{O}(\e|p^\e_*|^2 +|p^\e_*|^4)
\end{equation}
in an $\eps$-neighborhood of $p^\e_*$.

Reviewing the standard invariant manifold constructions by fixed
point/contraction arguments (see, e.g., \cite{Bressan, HK}), we find
that they yield existence and closeness in angle of these manifolds
to corresponding invariant subspaces of $dF^\e$ on a ball about
$p^\e_*$ for which the Lipshitz norm of the (total) nonlinear term
$N^\e+\Theta^\e$ is sufficiently small compared to the spectral gap
of $d\tilde{F}^\e(0)= dF^\e(p^\e_*)$ so long as the norm of $\log
d\tilde{F}^\e =\log dF^\e$ is no larger than some specified multiple
of the spectral gap, as it is here, cf. \eqref{DF}. Here the
spectral gap is defined as the minimum distance between one and the
modulus of eigenvalues that are not modulus one. By
\eqref{smalleig}, the spectral gap is greater than $\eta \eps$ for
some positive $\eta$, while by \eqref{Nbd}--\eqref{Thetabd} and the
fact that $|p^\e_*|\sim \eps^{1/2}$, we have
$|N^\e_z|+|\Theta^\e_z|\le C(|z|^2+\e^{1/2}|z| +\e^2)=o(\e)$ as
desired for $|z|<<\eps^{1/2}$. Thus, we obtain a detailed
``microscopic'' description of the behavior of the return map on a
ball about $p^\e_*$ of radius $\eta \eps^{1/2}$,  for $\eta>0$ and
sufficiently small. This translates to a detailed description of
$\cN_{p,\pm}$ and asymptotic behavior on an
$\eta\eps^{1/2}$-neighborhood of the unstable periodic orbit, far
more than what was needed (the excluded $\eps^{3/2}$-neighborhood in
$\omega^\e$ coordinates or for that matter the excluded $\e^2$
neighborhood in $\gamma^\e, A^\e$ coordinates) to complete the
argument.

Indeed, though we do not need it, the faster decay in $\eps$ of
perturbation terms $\Theta$, $\Theta_z$ yields uniform convergence
of perturbed to unperturbed flow and stable manifold in ``blown-up''
coordinates $\tilde \omega:=\omega/\eps^{1/2}$, by the same
argument.

%ENDCHANGED

%*******************************OLD VERSION***************
%
%Within the excluded $\eps^2$ neighborhood, on the other hand,
%corresponding to an $\eps^{3/2}$ neighborhood in $\beta$-coordinates,
%we may work in microscopic coordinates centered about the unperturbed
%unstable periodic orbits and rescaled by factor $\eps^{-1}$ in radial
%directions to obtain a bounded flow converging in $C^1$ as $\eps\to 0$ to
%the unperturbed flow, from which we find that, within a larger
%$\eps$-neighborhood $B_\eps$ in original, unrescaled $\beta$ coordinates,
%$\cN_{p,\pm}$ is order $\eps^{3/2}$ close to its unperturbed
%counterpart and uniformly repelling at order $\eps$ exponential
%rate throughout the $\eps$-neighborhood $B_\eps$; in particular, solutions
%in $B_\eps$ that are not on $N_{p,\pm}$ leave $\cK_\pm$ at
%$\eps$-exponential rate.
%We omit the details of this latter, essentially standard, analysis.
%
%*****************************************************

Thus, solutions originating outside $\cK_\pm$ remain outside, with $s$
by \eqref{seqn} strictly
decreasing at an exponential rate. Moreover, if ever solutions leave
$\cK_\pm$, they are attracted at an exponential rate to the
PT-symmetric manifold $\cN_{PT}$ corresponding to $s=0$.
On the other hand, solutions remaining in $\cK_\pm$ sufficiently
long must eventually enter the small neighborhood $B_\eps \cap \cK_\pm$,
after which solutions not on $\cN_{p,\pm}$ must
%CHANGED: added clarifying sentence here:
(by our microscopic description carried out in $\omega$-coordinates)
%ENDCHANGED
eventually leave.
Piecing together this information, we find that all solutions originating
in $\cMe\cap B(0,Ca)$ and not on the stable manifolds $\cN_{p,\pm}$
of the unstable periodics time-asymptotically approach the attracting
PT-symmetric manifold $\cN_{PT}$, as claimed.
This completes the description of asymptotic behavior within the
center manifold.

By Remark \ref{foliation}, the $C^\infty$ stable manifolds
$\cN_{p,\pm}$ within the center manifold, of codimension two
in $\cMe\cap B(0,a)$, extend to $C^\infty$ stable manifolds
$$
\tilde N_{p,\pm}:=\cup_{w\in \cN_{p,\pm}}P_\eps^{-1}(w)\cap B(0,\overline{C}a)
$$
of codimension two in $B(0,Ca)$.
Consequently, for fixed $\eps>0$, $(\tilde \cN_{p,+}\cup \tilde \cN_{p,-})\cap B(0,a)$
is exactly the set of data in $B(0,a)$ whose solutions converge
asymptotically to an unstable branch.
Likewise, there is a $C^\infty$ stable manifold
$\tilde \cN_{0}:=P_\eps^{-1}(0)$ of
the normal equilibrium, of codimension four, containing all solutions
originating in $B(0,a)$ and converging to the normal state.
%Since codimension two and four, these
%are null sets by any reasonable measure on $H^1$,
%hence nongeneric.
Finally, $\cN_{PT}$ has a $C^\infty$ stable manifold
$\tilde \cN_{PT}:=\cup_{w\in \cN_{PT}}P_\eps^{-1}(w)$ of codimension one.
By Proposition \ref{exp}, together with our description of asymptotic
behavior on the center manifold, we find that all solutions originating
in $B(0,a)$ outside the sets $\tilde \cN_{p,\pm}$ and $\tilde \cN_0$
are attracted to $\cN_{PT}$ (hence
to the larger manifold $\tilde \cN_{PT}$ of solutions converging
to $\cN_{PT}$) and
ultimately to the stable periodic solutions or their rotations.
This completes the description of asymptotic behavior and the proof.
\end{proof}

\begin{remark}\label{phaselock}
\rm{An implication of Proposition \ref{Hopf2} is that if one chooses
complex initial data $(\beta_1^\e(0),\beta_2^\e(0))$ away from the
special two-dimensional manifolds containing the unstable branch,
then there is a kind of phase-locking phenomenon whereby the flow
\eqref{bsys} pushes the solution towards a pair
$(\beta_1^\e(t),\beta_2^\e(t))$ which is a complex rotation of a
real pair, that is, $\gamma^\e$ tends to zero and the corresponding
solution to \eqref{reform} approaches a PT-symmetric profile. On the
other hand, it is not true that all solutions exhibit this
phase-locking, as evidenced by the existence of the unstable
periodic solutions.}
\end{remark}

%CHANGED, added:
\begin{remark}\label{validation}
\rm{
As noted parenthetically in the proof, a side-consequence of
our analysis is to rigorously validate the formal asymptotics
of Section \ref{NP} by verifying convergence of perturbed to
unperturbed phase portrait on a ball of size $\eps^{1/2}$ about
the normal state, in ``blown-up'' coordinates $\tilde \omega:=
\omega/\eps^{1/2}$ (equivalently, $\tilde \beta:=\beta/\eps^{1/2}$)
equivalent to the $\alpha$-coordinates of the earlier section.
}
\end{remark}
%ENDCHANGED

\subsection{Bifurcation off of higher eigenvalues} We conclude this
section by remarking on the bifurcation situation off of higher
eigenvalues. Recalling the notation \eqref{Lell}, we fix any integer
$\ell>1$, and set the parameter $\Gamma$ in \eqref{gl1} equal to
${\rm{Re}}\,\lambda_{2\ell-1}+\e$. Then \eqref{gl1}-\eqref{gl3} can
be rewritten as \beq
\psi_t=L_{\ell}[\psi]+\mathcal{N}(\psi,\e)\label{ellreform}\eeq

All of the analysis of the preceding subsection, including the
center manifold construction, applies to produce the existence of
stationary and periodic solutions to this system for $I<I_{\ell}$
and $I>I_{\ell}$ respectively, cf. \eqref{split}. The difference
here is that for $\ell>1$, there will always be a nonempty unstable
subspace corresponding to $\lambda_j$ with $1\leq j<2\ell-1$. Hence,
these bifurcating solutions will always be unstable. As far as the
center manifold construction is concerned, the primary change is
that the fixed point argument must now be applied to the integral
equation \bea &&y(x,t)=\Gamma(u,\e,y) :=e^{L_{1}t}u+\int_0^t
e^{L_{1}(t-\tau)}\Pi_c\,\mathcal{N}^{\delta}\left(y(x,\tau),\e\right)\,d\tau\nonumber\\
&&+ \int_{-\infty}^t
e^{L_{1}(t-\tau)}\Pi_s\,\mathcal{N}^{\delta}\left(y(x,\tau),\e\right)\,d\tau-
\int_{t}^{\infty}
e^{L_{1}(t-\tau)}\Pi_u\,\mathcal{N}^{\delta}\left(y(x,\tau),\e\right)\,d\tau,
\nonumber \eea rather than \eqref{fixed}. Here $\Pi_u$ denotes the
projection onto the unstable subspace and through its realization as
a contour integral \beq e^{tL_1}\Pi_u:=\int_{\Gamma_u}e^{\lambda
t}(\lambda I-L_1)^{-1}\,d\lambda, \label{bd2}\eeq (where $\Gamma_u$
is any bounded contour enclosing the finite number of eigenvalues
with positive real part) one establishes the necessary
bound\beq\norm{e^{tL_1}\Pi_u}_{H^1\to H^1}\leq Ce^{(\zeta_{1}-s)t}
\quad\mbox{for all}\;t\leq 0\nonumber \eeq to augment
\eqref{bd4}-\eqref{bd5}.

\section{Phase slip centers}

The physics literature associates periodic solutions with the
existence of phase slip centers (PSCs), that is, zeros of the order
parameter. Indeed, an immediate conclusion of the formal calculation
of Section \ref{NP}, rigorously confirmed in Section 6, is that the leading
order term in the expansion for the periodic solution $\psi_{p}$ to
\eqref{reform} along the bifurcation branch in the regime $I>I_c$
has a periodic array of zeros at $x=0$.  More precisely, referring
back to \eqref{stable2}-\eqref{ampe}, we see that the approximate
solution $\beta_1(t)v_1(x)+\beta_2(t)v_2(x)$, when evaluated at
$x=0$, has zeros whenever the expression
\[
\cos[\omega^{\e}t]:=\cos\left[\left(\imu+\frac{{\rm{Im}}\,\tilde{\chi}}{{\rm{Re}}\,\tilde{\chi}}\e+
\mathcal{O}(\e^{3/2})\right)t\right]
\]
vanishes, since $v_{2}(0)=i(u_{1}(0)-u_{2}(0))=0$. In other words,
there are PSC's located periodically at $(x,t)=(0,T^{\e}_{k})$,
$k=0,1,2,\ldots$ where $T^{\e}_{k}\approx
\frac{1}{\omega^{\e}}(\pi/2+k\pi)$. Since the actual solution
$\psi_{p}$ is uniformly close to $\beta_1v_1+\beta_2v_2$, it follows
that the two functions have the same Brouwer degree in a
neighborhood of the points $\{(0,T^{\e}_{k})\}$. Hence, we will
rigorously conclude that $\psi_{p}$ possesses an array of zeros
close to the points $\{(0,T^{\e}_{k})\}$ once we check that the
degree of $\beta_1v_1+\beta_2v_2$ is nonzero at these points. But a
direct calculation of the Jacobian determinant at these points
indeed reveals that
\[
{\rm{Jac}}\,\left(\beta_1v_1+\beta_2v_2\right)(0,T^{\e}_k) \sim
-\frac{8\pi\e}{\kappa_{r}\imu+\kappa_{i}\e}{\rm{Re}}\,u_{1}'(0).
\]
A straight-forward numerical calculation
of the first eigenfunction for the operator $M$ (cf. \eqref{gl13}),
for which there are rigorous error bounds, shows that
${\rm{Re}}\,u_{1}'(0)\not=0$, with values ranging monotonically from about $-0.2234$
for $I=12.5$ down to about $-0.3578$ for $I=20.$
Hence, we obtain a rigorous
confirmation of a periodic array of PSC's for the periodic solution.

\end{document}